\documentclass[12pt]{iopart} 
\usepackage{iopams}
\usepackage{setstack}
\usepackage{amstext}
\usepackage{amsgen}
\usepackage{amsfonts}
\usepackage{amsbsy}
\usepackage{amsopn}
\usepackage{amsthm,amssymb,upref,amscd}
\begin{document}

\newtheorem{theorem}{Theorem}
\newtheorem{lemma}[theorem]{Lemma}
\newtheorem{corollary}[theorem]{Corollary}
\newtheorem{definition}[theorem]{Definition}
\newcommand{\binom}[2]{{#1 \choose #2}}

\hyphenation{schwarz-schild hyper-geo-metric}

\title[Static spacetimes with prescribed multipole moments]{Static spacetimes with prescribed multipole moments; a proof of a conjecture by Geroch}
\author{Magnus Herberthson}
\address{Department of  
Mathematics, Link\"oping University,
SE-581 83 Link\"oping, Sweden.
e-mail:  maher@mai.liu.se}
\begin{abstract}
In this paper we give sufficient conditions on a sequence of multipole moments
for a static spacetime to exist with precisely these moments.
The proof is constructive in the sense that a metric
having prescribed multipole moments up to a given order can be calculated. Since these sufficient conditions agree with already known necessary conditions, this completes the proof of a long standing conjecture due to Geroch. 
\end{abstract}

\section{Introduction} \label{intro}
There are various definitions of multipole moments, \cite{quevedo}. For static asymptotically flat spacetimes, the standard definition was given by Geroch \cite{geroch}, and this definition was later extended to the stationary case by Hansen \cite{hansen}. In essence, the recursive definition by Geroch from \cite{geroch}, is a non-trivial modification of the Newtonian definition taking place in a conformally rescaled spacetime, see Section \ref{tmoments}. Since the conformal factor involved is not unique the relativistic definition must, apart from taking into account the curvature terms, ensure invariance under the conformal freedom available.

In \cite{geroch}, Geroch made two conjectures, namely ($i^0$ will be defined shortly) \medskip \\
{\em Conjecture1: } Two static solutions of Einstein's equations having identical multipole moments coincide, at least in some neighborhood of $i^0$.
 \medskip \\
and
 \medskip \\
{\em Conjecture 2:} Given any set of multipole moments, subject to the  appropriate convergence condition, there exists a static spacetime of Einstein's equations having precisely those moments.
 \medskip \\
Conjecture 1 was proved in \cite{beig} and although many properties of the multipoles are known, see e.g. \cite{quevedo}, Conjecture 2 has not yet been proven. There are partial results however.
For instance, it is known \cite{simon} that an arbitrary sequence of multipole moments uniquely defines formal power series of relevant field variables, and that if the series converge, this give a static spacetime having these moments. In \cite{friedrich}, on the other hand, growth conditions are given
on a set referred to as {\em null data}, which ensures the existence of a static space time possessing these null data. The null data are shown to be in a one-to-one correspondence with the Geroch multipole moments (with non-zero monopole), but since this correspondence
is rather implicit, no growth condition on the actual multipole moments are derived. On the "necessary side", the situation is more satisfactory. In \cite{bahe}, it was shown that the multipole moments of any asymptotically flat stationary, and therefore static, spacetime  do not grow "too fast", and precise conditions were given.

The purpose of this paper is to show that the necessary conditions in \cite{bahe} are also
sufficient, i.e., if they are satisfied, there exists a
static asymptotically flat vacuum spacetime with these moments. Therefore, this will prove Geroch's Conjecture 2 above. The proof in Section \ref{majorsection} will include an explicit
recursion of the desired metric, which means that the metric for a static spacetime with prescribed moments up to an arbitrary order can be calculated. We will first prove that the metric cast in a special form is uniquely defined by the moments, at least formally, and then show that the result from \cite{friedrich} can be used to deduce convergence of the series for the metric components. 

One major obstacle has been to explicitly link a certain metric and potential with a given set of multipole moments. This problem goes back to the actual definition of the moments, i.e., the recursion (\ref{orgrec}), where the operation of "taking the totally symmetric and trace-free part" effectively obscures the relation between the metric, the potential and the moments, unless some care is taken. 
Another issue is the coordinate freedom. Without extra restrictions
on the metric it is impossible to uniquely derive the metric components from the moments
since one can always change coordinates and hence the components of the metric.
However, as will be shown in the following sections, these two issues can be addressed simultaneously.

We also remark that the sufficient conditions given do not contain the usual requirement that the mass term (monopole) $m$ is non-zero, although that is a condition which may be required for physical reasons.

\section{Multipole moments of stationary spacetimes} \label{tmoments}

Although in this paper, we  address a conjecture concerning static spacetimes, we have chosen
to formulate the conjecture within the stationary setting. One reason is that we will refer to results from \cite{bahe}, where this setting is used. 
Therefore,  in this section we quote the definition of multipole moments given by Hansen in \cite{hansen}, which is an extension 
to stationary spacetimes of the definition for static spacetimes by Geroch \cite{geroch}.
We thus consider a stationary spacetime $(M, g_{ab})$ with signature $(-,+,+,+)$ and  with a timelike Killing vector field $\xi^a$.
We let $\lambda =-\xi^a\xi_a$ be the norm, and define the twist $\omega$ through 
$\nabla_a\omega=\epsilon_{abcd}\xi^b\nabla^c\xi^d$.
If $V$ is the 3-manifold of trajectories, the metric $g_{ab}$ induces the 
positive definite metric
$$h_{ab}=\lambda g_{ab}+\xi_a\xi_b$$ on $V$.
It is required that $V$ is asymptotically flat, i.e.,  there exists a 3-manifold
$\hat V$ and a conformal factor $\Omega$ satisfying
\begin{itemize}
\item[(i)]{$\hat V = V \cup i^0$, \; where $i^0$ is a single point}
\item[(ii)]{$\hat h_{ab}=\Omega^2 h_{ab}$ is a smooth metric on $\hat V$}
\item[(iii)]{At $i^0$, $\Omega=0, \hat D_a \Omega =0, \hat D_a \hat D_b \Omega = 2 \hat h_{ab}$,}
\end{itemize}
where $\hat D_a$ is the derivative operator associated with $\hat h_{ab}$.
$i^0$ is referred to as spacelike infinity\footnote{$i^0$ is also used in a four-dimensional context.}.
On $M$, and/or $V$, one defines the scalar potential
\begin{equation} \label{hansenpot}
\phi=\phi_M+i\phi_J, \quad \phi_M=\frac{\lambda^2+\omega^2-1}{4\lambda}, \,\phi_J=\frac{\omega}{2\lambda}.
\end{equation}
The multipole moments of $M$ are then defined on $\hat V$ as certain  
derivatives of the scalar
potential $\hat \phi=\phi/\sqrt \Omega$ at $i^0$. More  
explicitly, following \cite{hansen}, let $\hat R_{ab}$ denote
the Ricci tensor of $\hat V$, and let $P=\hat \phi$. Define the  
sequence $P, P_{a_1}, P_{a_1a_2}, \ldots$
of tensors recursively:
\begin{equation} \label{orgrec} P_{a_1 \ldots a_n}=C[
\hat D_{a_1}P_{a_2 \ldots a_n}-
 \frac{(n-1)(2n-3)}{2}\hat R_{a_1 a_2}P_{a_3 \ldots a_n}],
\end{equation}
where $C[\ \cdot \ ]$ stands for taking the totally symmetric and  
trace-free part. The multipole moments
of $M$ are then defined as the tensors $P_{a_1 \ldots a_n}$ at  
$i^0$. 

In \cite{geroch}, a slightly different setup and a different potential is used, but it is known \cite{simon} that the potential used there and  (\ref{hansenpot}) with $\omega=0$
produce the same moments.

\section{Static spacetimes with prescribed multipole moments} \label{majorsection}
In this section, we will formulate and prove the desired theorem. The theorem will be as conjectured
in \cite{bahe}, i.e., in essence that if the sequence of multipole moments is naturally connected to a harmonic function on $\mathbf{R}^3$, there exists a static asymptotically flat vacuum spacetime having precisely those moments. Namely, (cf. Theorem 8 in \cite{bahe}), consider $\mathbf{R}^3$ with Cartesian coordinates 
$\mathbf{r}=(x,y,z)=(x^1, x^2, x^3)$, and let $\alpha=(\alpha_1,\alpha_2,\alpha_3)$ be a multi-index. With the convention that, in terms of components,  
$P^0_\alpha=P^0_{\stackrel{\underbrace{\scriptstyle {11\cdots1}}}{\alpha_1} 
\stackrel{ \underbrace{{\scriptstyle 22\cdots 2}}}{\alpha_2} 
\stackrel{\underbrace{{\scriptstyle 33\cdots 3}}}{\alpha_3} }$, we have the following theorem.
\begin{theorem} \label{thm}
Let $P^0, P^0_{a_1}, P^0_{a_1 a_2}, \ldots$ be a sequence of real valued totally symmetric and trace free tensors on $\mathbf{R}^3$, and let $P^0, P^0_{i_1}, P^0_{i_1 i_2}, \ldots$ be the corresponding components with respect to the Cartesian coordinates $\mathbf{r}=(x,y,z)$. If $u(\mathbf{r})=\sum_{|\alpha|\geq 0}\frac{\mathbf{r}^\alpha}{\alpha!}
P^0_\alpha $ converges in a neighbourhood of  the origin in $\mathbf{R}^3$, there exists a static asymptotically flat vacuum spacetime having the moments $P^0, P^0_{a_1}, P^0_{a_1 a_2}, \ldots$.
\end{theorem}
Note that we do not require the monopole $P^0=\hat  \phi(0)$ to be non-zero\footnote{Due to the definition of $\phi$, $P^0=-m$, where $m$ is the mass of the spacetime.}. The proof, however, will first
be carried out under the assumptions $P^0 \neq 0$, and this condition will then be relaxed in Section \ref{massless}.
The proof for the case $P^0 \neq 0$ will be carried out in a sequence of lemmas (Lemma \ref{r2} - Lemma \ref{f2}), which also 
show that the metric up to a given order can be calculated explicitly, but first we will give an outline of the proof, discuss the notation and formulate the relevant field equations.
\subsection{Outline of the proof}
As mentioned in Section \ref{intro}, it is possible to reduce the coordinate freedom in the metric
components, and simultaneously establish a direct link between the potential $\hat \phi$ and
the desired moments. This will be done in Section \ref{prescribed}, where the link is expressed 
in Theorem \ref{r2thm}.

In Section \ref{confeqsec}, we will address the conformal field equations from 
\mbox{Section \ref{tfe}}, and see that the form of the metric
given in Section \ref{prescribed} results in singular equations. By requiring that the conformal field
equations are smooth at $i^0$, further restriction will be put on the rescaled metric $\hat h_{ij}$,
which then takes its final form.

With the form of the metric fixed, we will in Section \ref{unique} show that the field equations
determine the metric components as a formal power series. More precisely, we will show that
when the monopole is non-vanishing,
a certain subset of the field equations is sufficient to specify the metric completely.

In Section \ref{conv} we will address the full set of equations, as well as the issue of
convergence of the power series derived. By referring to a result by 
\mbox{Friedrich, \cite{friedrich}}, convergence of the power series will be concluded.
It will also be seen that the full set of equations are satisfied.

Finally, in Section \ref{massless}, we will relax the condition that the monopole is nonzero.

\subsection{Notation} \label{notation}
Small Latin letters $a,b,\ldots$ refer to abstract indices, as in Section \ref{tmoments}. Small Latin letters
$i,j,k, \ldots$ are numerical indices and refer to components with respect to the normal coordinates
$(x,y,z)=(x^1,x^2,x^3)$ introduced below. Since these components refer to this particular coordinate system only, the equations will not be tensor equations. With this said, we will still use $=$ instead of
$\dot =$.

Almost all variables\footnote{The exception is $r=\sqrt{x^2+y^2+z^2}$} are assumed to be formally analytic, i.e., they admit a formal power series expansion
(again in terms of the chosen coordinates), so that, for a tensor, $A_{ijk}$ say, we can write
$$
A_{ijk}=\sum_{n=0}^\infty A_{ijk}^{[n]}
$$
where $A_{ijk}^{[n]}$ denots polynomials in $(x,y,z)$ which are homogeneous of order $n$
(and where the summation may be formal). Both  $\eta_{ij}$ and $\eta^{ij}$ denotes the identity
matrix, and by $[A_{ij}]$ we denote the trace of $A_{ij}$, i.e., $\eta^{ij}A_{ij}$. We also use $\partial_i=\frac{\partial}{\partial x^i}.$

With $r=\sqrt{x^2+y^2+z^2}$, $f_1 \equiv f_2 \pmod{r^2}$ means that $f_1(x,y,z)-f_2(x,y,z)=
r^2\ g(x,y,z)$ for some (formally analytic) function $g$. When $f \equiv 0 \pmod{r^2}$, so that
$f=r^2\ g$ for some formally analytic function $g$, we also use the shorter notation $r^2 | f$.

In the proof of Lemma \ref{f2} we will refer to \cite{friedrich} and hence use some of the notation there.

\subsection{The field equations} \label{tfe}
Apart from having the correct multipole moments, we must ensure that the metric describes
a static vacuum spacetime. We will formulate our equations on the 3-manifold 
$(\hat V, \hat h_{ab})$ defined in Section \ref{tmoments}, starting with the field equations from \cite{geroch}.
However, the 3-manifolds in \cite{geroch} and \cite{hansen} are defined slightly differently, and the relations imply the following.

Starting with a static spacetime $(M,g_{ab})$ with timelike killingvector $\xi^a$, we put
$0<\lambda=-\xi^a \xi_a$ and  $\Psi=1-\sqrt{\lambda}$. 
From \cite{geroch} we consider\footnote{In \cite{geroch}, the potential $\psi=-\Psi$ is used.}
a 3-surface $V_G$ orthogonal to $\xi^a$. In terms of the metric on $V_G$: $(h_G)_{ab}=g_{ab}+\xi_a \xi_b/\lambda $, the field equations are, \cite{geroch},
\begin{equation}\label{gereq}
\left\{
\begin{array}{c} (D_G)^a (D_G)_a \Psi=0 \\ 
(R_G)_{ab}=\frac{1}{\Psi-1}(D_G)_a (D_G)_b \Psi \end{array}\right.
\Leftrightarrow
\left\{
\begin{array}{c} {(R_G)^a}_a=0 \\ 
(R_G)_{ab}=\frac{1}{\sqrt \lambda}(D_G)_a (D_G)_b \sqrt \lambda \end{array}\right.,
\end{equation}
where $(D_G)_a$ is the derivative operator and $(R_G)_{ab}$ is the Ricci tensor associated
with $(V_G,(h_G)_{ab})$.

On the 3-manifold $V$ on the other hand, the metric is $h_{ab}=\lambda g_{ab}+\xi_a\xi_b$, i.e., 
$h_{ab}=\lambda (h_G)_{ab}$, and with $\hat h_{ab}=\Omega^2 h_{ab}$, this implies that
$\hat h_{ab}=(\sqrt{\lambda}\Omega)^2 (h_G)_{ab}$. We can now express equations 
(\ref{gereq}) on $(\hat V, \hat h_{ab})$ using as conformal factor $\hat \Omega=\sqrt{\lambda} \Omega$.
Using the properties of conformal transformations, \cite{wald}, we find that (\ref{gereq})
becomes 
\begin{equation}\label{confeq}
\begin{array}{lcc}
\hat R+4\hat h^{ab}\hat D_a \hat D_b \ln (\sqrt{\lambda} \Omega)-
2\hat h^{ab}\hat D_a \ln (\sqrt{\lambda} \Omega)\hat D_b \ln (\sqrt{\lambda} \Omega) & = & 0\\
\hat R_{ab}+\hat D_a \hat D_b  \ln (\sqrt{\lambda} \Omega)+\hat h_{ab}\hat h^{cd}\hat D_c \hat D_c  \ln (\sqrt{\lambda} \Omega) +& & \\
\hat D_a \ln (\sqrt{\lambda} \Omega)\hat D_b \ln (\sqrt{\lambda} \Omega)-\hat h_{ab}\hat h^{cd}\hat D_c \ln (\sqrt{\lambda} \Omega)\hat D_c \ln (\sqrt{\lambda} \Omega) &=& \\
\frac{1}{\sqrt{\lambda}}[\hat D_a \hat D_b \sqrt \lambda+\hat D_a\ln(\sqrt \lambda \Omega)\hat D_b \sqrt \lambda
+& &\\
 \hat D_b\ln(\sqrt \lambda \Omega)\hat D_a \sqrt \lambda-
\hat h_{ab} \hat h^{cd}\hat D_d\ln(\sqrt \lambda \Omega)\hat D_c\sqrt{\lambda}].
\end{array}
\end{equation}
Thus, we are looking for a metric $\hat h_{ab}$ on $\hat V$, defined in a neighbourhood of $i^0$, which
satisfies (\ref{confeq}) and where the corresponding 4-dimensional spacetime has prescribed multipole moments.
As mentioned earlier, the equation (\ref{confeq}) cannot have a unique solution  in terms of coordinates.
This is not only due to the fact that one can always represent a metric in different coordinate systems, but also
because there is a freedom in the conformal factor $\Omega$. In the next section we will put the metric $\hat h_{ab}$
in a canonical form, i.e., use a preferred coordinate system. This coordinate system is also constructed so that
it allows specification of the multipole moments, i.e., puts (\ref{orgrec}) in a form which is more transparent.

\subsection{Prescribed moments and canonical form of the metric} \label{prescribed}
In this section we will, in a sense, reverse the arguments from \cite{bahe}.
The key point is to work in normal coordinates, in which the Ricci tensor $\hat R_{ab}$
takes a special form. This will lead to Theorem \ref{r2thm}, which allows for a direct connection
between the potential $P=\hat \phi$ and the multipole moments.

Let us therefore first introduce normal coordinates $(x,y,z)=(x^1,x^2,x^3)$
on $\hat V$, where the point $i^0$ has coordinates $\mathbf 0$, and such that $(x,y,z)$
"are Cartesian" at $i^0$, i.e., in terms of the coordinates, $\hat h_{ij}(i^0)=\eta_{ij}$.
By assumption, $\hat h_{ab}$, (and therefore $\hat R_{ab},\ \hat R$) and also the
conformal factor $\Omega$ are at this stage formal power series, while  $u$ from
Theorem \ref{thm} is given by a power series which converges in a neighbourhood of $\mathbf 0$.
$u$ is defined on $\mathbf{R}^3$, but in terms of the normal coordinates introduced, the sum $\sum_{|\alpha|\geq 0}\frac{\mathbf{r}^\alpha}{\alpha!}
P^0_\alpha $ can also be interpreted on $\hat V$, namely as the potential function
$\hat \phi(\mathbf{r})$.

We now turn to the condition on $\hat R_{ab}$. The condition is in \cite{bahe} formulated through the complexification $\hat V_\mathbf C$ of $\hat V$, and by using complex null vectors $\eta^a$, i.e., vectors
$\eta^a$ with $\eta^a \eta_a=0$. Using normal coordinates, it was shown in \cite{bahe} that
for any fixed $\varphi$, the complex curve
$$\gamma_\varphi: t \to (t \cos\varphi,t \sin \varphi, i t),$$
has tangent vector
$\eta^a=\eta^a_\varphi(t)=
\cos\varphi (\frac{\partial}{\partial x^1})^a+\sin\varphi (\frac{\partial}{\partial x^2})^a+i (\frac{\partial}{\partial x^3})^a$ which satisfies $\eta^a \eta_a=0$. The relevant condition on $\hat R_{ab}$ is then
to be found in Lemma 6 of \cite{bahe}, where the condition $\tilde \eta^a \tilde \eta^b \tilde R_{ab}=0$ was made.
With our notation this reads $\eta^a \eta^b \hat R_{ab}=0$, and it is this condition, together with the normality of the coordinates, which allows for a direct connection between the potential $\hat \phi$
and the moments (Theorems 7 and 8 in \cite{bahe}, Theorem \ref{r2thm} below). Although the results in \cite{bahe}, which build
on the techniques developed in \cite{herb,bahe1,bahe2}, depend crucially on complex quantities
(especially the concept of "leading term") it is interesting to note that all of these tools and arguments can be given
in purely real terms. Thus, rather then referring to, and reversing, the arguments in \cite{bahe},
we will derive/translate the corresponding conclusion using only real quantities. 

The first property is the analogue of Lemma 1c in \cite{bahe}, namely, with $\eta^a$ as above, that
$\eta^{a_1} \ldots \eta^{a_n}T_{a_1 \ldots a_n}=\eta^{a_1} \ldots \eta^{a_n}C[T_{a_1  \ldots a_n}]$,
although we here state it in a vector space using the radius vector.

\begin{lemma} \label{r2}
Let $(x,y,z)$ be Cartesian coordinates on $\mathbf{R}^3$, let $\mathbf{r}^a=x (\frac{\partial}{\partial x})^a+y(\frac{\partial}{\partial y})^a+z (\frac{\partial}{\partial z})^a
 $, and put $r=\sqrt{x^2+y^2+z^2}$.
Then, for any tensor $T_{a_1  \ldots a_n}$,
$\mathbf{r}^{a_1} \ldots \mathbf{r}^{a_n}C[T_{a_1  \ldots a_n}]\equiv
\mathbf{r}^{a_1} \ldots \mathbf{r}^{a_n}T_{a_1  \ldots a_n} \pmod {r^2}$
\end{lemma}
\begin{proof}
$C[T_{a_1  \ldots a_n}]$ is constructed from $T_{a_1  \ldots a_n}$ by taking the totally symmetric
part and subtracting suitable amounts of the tensors
$\eta_{(a_1 a_2}T_{a_3  \ldots a_n) \phantom{b} b}^{\phantom{a_3  \ldots a_1} b}$, 
$\eta_{(a_1 a_2}\eta_{a_3 a_4}T_{a_5  \ldots a_n) \phantom{b} b \phantom{b} c}^{\phantom{a_3  \ldots a_1} b \phantom{b} c}$,
$\ldots$ where $\eta_{ab}$ is the Euclidean metric on $\mathbf R^3$. Thus, it is clear that
$C[T_{a_1  \ldots a_n}]=T_{(a_1  \ldots a_n)}+\eta_{(a_1 a_2}S_{a_3 a_4 \ldots a_n)}$ for some tensor
$S_{a_3 a_4 \ldots a_n}$.
Since $\mathbf{r}^{a}\mathbf{r}^{b} \eta_{ab}=r^2$, the statement follows by transvecting with
$\mathbf{r}^{a_1} \ldots \mathbf{r}^{a_n}$ and observing that 
$\mathbf{r}^{a_1} \ldots \mathbf{r}^{a_n}T_{(a_1  \ldots a_n)}=\mathbf{r}^{a_1} \ldots \mathbf{r}^{a_n}T_{a_1  \ldots a_n}$.
\end{proof}
This property can also be realized in the following way. It is clear that
$\mathbf{r}^{a_1} \ldots \mathbf{r}^{a_n}T_{a_1  \ldots a_n}$ is a homogeneous polynomial $p=p(x,y,z)$ of degree $n$.
Using spherical coordinates in $\mathbf R^3$, i.e., 
$x=r \sin\theta \cos \phi$, $y=r \sin\theta \sin \phi$, $z=r \cos\theta$, we 
express $p$ in spherical harmonics $Y_l^m$ and get
$p=\sum \limits_{l=0, l \mbox{ even }}^{n}\sum \limits_{m=-l}^l r^n c_l^m Y_l^m(\theta,\phi)$ if $n$ is even.
If $l=n-2$ or less, $\sum \limits_{m=-l}^l r^n c_l^m Y_l^m(\theta,\phi)=r^{n-l} \sum \limits_{m=-l}^l r^{l} c_l^m Y_l^m(\theta,\phi)$
with $r^{n-l}$ smooth and divisible by $r^2$, and where the sum is a smooth polynomial.  
This means that $\mathbf{r}^{a_1} \ldots \mathbf{r}^{a_n}C[T_{a_1  \ldots a_n}]$ corresponds to the polynomial
$p_C=\sum \limits_{m=-n}^n r^n c_n^m Y_n^m(\theta,\phi)$, i.e., only the terms corresponding to $l=n$ are kept.
Identical remarks hold if $n$ is odd. 

Using normal coordinates, the property of Lemma \ref{r2} can be carried over to the manifold $\hat V$. Namely, from \cite{thomas}
we have the following lemma, which we formulate for the particular case when $\hat h_{ij}(i^0)=\eta_{ij}$.
\begin{lemma} \label{normal}
$(x,y,z)=(x^1,x^2,x^3)$ are normal coordinates on $(\hat V,\hat h_{ab})$, if and only if 
$x^i \hat h_{ij}(x^k)=x^i \hat h_{ij}(i^0)$, i.e., $x^i \hat h_{ij}=x^i \eta_{ij}$.
\end{lemma}
Thus, in normal coordinates, the above lemma implies that
$x^i x^j \hat h_{ij}=x^2+y^2+z^2=r^2$. This means that by the same arguments as above, 
\begin{equation} \label{ctv}
x^{i_1} \ldots x^{i_n}C[T_{i_1  \ldots i_n}]\equiv
x^{i_1} \ldots x^{i_n}T_{i_1  \ldots i_n}  \pmod {r^2} \mbox{ in } \hat V
\end{equation} 
Remark. That (\ref{ctv})  is still true in $\hat V$ depends on the fact that
we have $x^i x^j \hat h_{ij}=r^2$. However, in Lemma \ref{r2},  we had a fixed tensor
$T_{a_1  \ldots a_n}$ on a vector space, while on $\hat V$, $T_{a_1  \ldots a_n}$ is a tensor
field. This does not affect the equality (\ref{ctv}), but when evaluating at $i^0$,
where $(x^1,x^2,x^3)=(0,0,0)$, it just says $0=0$. Another effect is that
$x^{i_1} \ldots x^{i_n}T_{i_1  \ldots i_n}$ is a polynomial which is the sum of
a homogeneous polynomial of degree $n$ and a polynomial containing only higher order terms. 
However, by replacing $x^i$ with $\rho^i=x^i/r$, so that $\rho^i \rho^j \hat h_{ij}=1$,
both  $\rho^{i_1} \ldots \rho^{i_n}C[T_{i_1  \ldots i_n}]$ and $\rho^{i_1} \ldots \rho^{i_n}T_{i_1  \ldots i_n}$ will be direction dependent quantities at $i^0$. In the limit $r\to 0$, the higher order terms
vanish and we get
\begin{equation*}
\begin{array}{ccl}
\rho^{i_1} \ldots \rho^{i_n}T_{i_1  \ldots i_n}(i^0) & = &
\sum \limits_{l=0, l \mbox{ even }}^{n}\sum \limits_{m=-l}^l c_l^m Y_l^m(\theta,\phi) \Rightarrow \\
 \rho^{i_1} \ldots \rho^{i_n}C[T_{i_1  \ldots i_n}(i^0)] &=&
\sum \limits_{m=-n}^n c_n^m Y_n^m(\theta,\phi),
\end{array}
\end{equation*}
if $n$ is even; and a corresponding relation when $n$ is odd ($l=1,3,5, \ldots n$). This last equality
also tells us that $C[T_{i_1  \ldots i_n}]$ transvected with $x^{i_1} \ldots x^{i_n}$ or
$\rho^{i_1} \ldots \rho^{i_n}$ still contains the full information, since there are $2n+1$
degrees of freedom in the RHS.

In view of Lemma \ref{r2} and the corresponding property (\ref{ctv}) on $\hat V$, it is obvious that $\mathbf{r}^{a}\mathbf{r}^{b}\hat R_{ab} \equiv \ 0 \pmod {r^2}$, i.e.,
that $r^2 | \mathbf{r}^{a}\mathbf{r}^{b}\hat R_{ab}$ is a desirable property in order to simplify (\ref{orgrec}) since $\mathbf{r}^{a}\mathbf{r}^{b}\hat R_{ab}$ will then not affect the multipole
moments.  As it turns out, this condition can be formulated purely in algebraic terms. This will be done in Lemma \ref{r2Rab} below.
To prepare for this lemma, we need some more tools. 

\begin{lemma} \label{gdec}
Suppose $A=A_{ij}=A_{(ij)}$ are the components of a symmetric tensor field $A_{ab}$ on $\hat V$ with respect to normal coordinates $(x,y,z)=(x^1,x^2,x^3)$
 and that $A_{ij}$ has the property that $r^2|\eta^{ij}A_{ij}$, $A_{ij}x^i=0$. Then $A$ is uniquely decomposable as 
$$
\begin{array}{rcl}
A&=&f_1(x,y) A_1+f_2(y,z) A_2+f_3(x,z) A_3+f_4(x,z) A_4\\
&&+f_5(x,y,z) A_5+f_6(x,y,z) A_6+f_7(x,y,z) A_7, \mbox{ where}
\end{array}
$$
$
\begin{array}{rcl}
A_1&=&\left(
\begin{array}{lll}
 -y^2-z^2 & x y & x z \\
 x y & -x^2-z^2 & y z \\
 x z & y z & -x^2-y^2
\end{array}
\right)\\
A_2&=&\left(
\begin{array}{lll}
 2 z \left(y^2+z^2\right) & -x y z & -x \left(y^2+2 z^2\right) \\
 -x y z & 0 & x^2 y \\
 -x \left(y^2+2 z^2\right) & x^2 y & 2 x^2 z
\end{array}
\right)\\
A_3&=&\left(
\begin{array}{lll}
 2 x z^2 & y z^2 & -\left(2 x^2+y^2\right) z \\
 y z^2 & 0 & -x y z \\
 -\left(2 x^2+y^2\right) z & -x y z & 2 x \left(x^2+y^2\right)
\end{array}
\right)\\
A_4&=&\left(
\begin{array}{lll}
 0 & -x y z & x y^2 \\
 -x y z & 2 z \left(x^2+z^2\right) & -y \left(x^2+2 z^2\right) \\
 x y^2 & -y \left(x^2+2 z^2\right) & 2 y^2 z
\end{array}
\right)\\
A_5&=&\left(
\begin{array}{lll}
 0 & -z \left(y^2+z^2\right) & y \left(y^2+z^2\right) \\
 -z \left(y^2+z^2\right) & 2 x y z & -x (y-z) (y+z) \\
 y \left(y^2+z^2\right) & -x (y-z) (y+z) & -2 x y z
\end{array}
\right)\\
A_6&=&\left(
\begin{array}{lll}
 2 x y z & -z \left(x^2+z^2\right) & -y (x-z) (x+z) \\
 -z \left(x^2+z^2\right) & 0 & x \left(x^2+z^2\right) \\
 -y (x-z) (x+z) & x \left(x^2+z^2\right) & -2 x y z
\end{array}
\right)\\
A_7&=&\left(
\begin{array}{lll}
 0 & x z^2 & -x y z \\
 x z^2 & 2 y z^2 & -\left(x^2+2 y^2\right) z \\
 -x y z & -\left(x^2+2 y^2\right) z & 2 y \left(x^2+y^2\right)
\end{array}
\right).
\end{array}
$
\end{lemma}

\begin{proof}
In terms of matrices, the condition $r^2|\eta^{ij}A_{ij}$ is just that $r^2$ divides the trace of $S$, i.e.,
$r^2 | [A]$, or $[A] \equiv 0 \pmod {r^2}$. Therefore, we can make the following ansatz
\begin{equation*}  \nonumber 
A=\left(
\begin{array}{ccc}
X(x,y,z) & a(x,y,z) & b(x,y,z)\\
a(x,y,z) & Y(x,y,z) & c(x,y,z)\\
b(x,y,z) & c(x,y,z) & r^2 Z(x,y,z)-X(x,y,z)-Y(x,y,z)
\end{array}
 \right)
\end{equation*}
where $a,b,c,X,Y$ and $Y$ are analytic in the variables indicated. The condition $A_{ij}x^i=0$
then translates to
\begin{equation}\label{sref}
\hspace*{-10mm}\begin{array}{ccc}
x X(x,y,z) +y  a(x,y,z) +z b(x,y,z) &=&0\\
x a(x,y,z) +y Y(x,y,z) +z c(x,y,z)&=&0\\
x b(x,y,z) +y  c(x,y,z) +z[ r^2 Z(x,y,z)-X(x,y,z)-Y(x,y,z)] &=& 0
\end{array}
\end{equation}
We use analyticity of the functions involved, and start by writing
\begin{equation*}  \nonumber
\begin{array}{ccl}
X(x,y,z) &
 = & X_1(x) + y\, X_2(x) + y^2 X_3(x, y) + z\, X_4(x, y, z),\\  
Y(x,y,z) &
 = & Y_1(y) + x\, Y_2(y) + x^2 Y_3(x, y) + z\, Y_4(x, y, z),\\ 
a(x,y,z) &
 = & a_1 +x\, a_2(x) + y\, a_3(y) + x y\, f_1(x, y) + z\, a_4(x, y, z).
\end{array}
\end{equation*}
Inserted in (\ref{sref}), the limits $z \to 0, y \to 0$ gives $X_1(x) = 0,a_1=0, a_2(x)
= 0$, while
the limits $z \to 0, x \to 0 $ gives $Y_1(y) = 0, a_3(y) = 0$. Given this, $z \to 0$ implies
$X_2(x) = 0, Y_2(y) = 0$, followed by $X_3(x,y)=-f_1(x,y)$ and  $Y_3(x,y)=-f_1(x,y)$.
A further evaluation of (\ref{sref}) then shows that
$b(x,y,z) = -x\, X_4(x, y, z) - y\, a_4(x, y, z)$ and
$ c(x,y,z) = -x\, a_4(x, y, z) - y\, Y_4(x, y, z)$.
By writing 
$$Y_4(x, y, z) = Y_5(x, z) + y\, Y_6(x, y, z) - z\,  f_1(x, y),$$
 (\ref{sref}) together with
$y=0$ shows that $Y_5(x,z)$ contains the factor $x^2+z^2$; and similarly
$$X_4(x, y, z) = X_5(y, z) + x\, X_6(x, y, z) - z\, f_1(x, y)$$ in (\ref{sref}) with
$x=0$ reveals that $X_5(y,z)$ contains $y^2+z^2$. Thus
$Y_5(x, z) = 2(x^2 + z^2) f_4(x, z)$ and $X_5(y, z) = 2(y^2 + z^2) f_2(y, z)$.
To proceed, we write 
$$X_6(x, y, z) = X_8(x, z) + 2y\, f_6(x, y, z).$$ 
In addition, (\ref{sref}) with $y=0$ then shows that $X_8(x, z) = 2z\, f_3(x, z)$ for some function $f_3$.
Next, (\ref{sref}) implies 
$Z(x,y,z) = 2y\, f_7(x, y, z) -2 f_1(x, y) + 2 x\, f_3(x, z) + 2 z\, f_2(y, z)+ 2 z\, f_4(x, z)$, for some function
$f_7$. Also, 
(\ref{sref}) with $x=0$ then gives $Y_6(x, y, z) = 2 z\, f_7(x, y, z) + 2x\, f_5(x, y, z)$, and a final
application of (\ref{sref}) shows that 
$a_4(x, y, z) = y\ z\ f_3(x, z) - x\ y\ f_2(y, z) - x^2\ f_6(x, y, z) - 
z^2\ f_6(x, y, z) - x\ y\ f_4(x, z) - y^2\ f_5(x, y, z) - 
z^2\ f_5(x, y, z) + x\ z\ f_7(x, y, z).
$
By collecting terms, we get Theorem \ref{gdec}.
\end{proof}
In the proof of Lemma \ref{r2Rab} below, $A_{ij}$ will be the difference $\hat h_{ij}-\eta_{ij}$,
and since $\hat R_{ij}$ involves the Christoffel symbols, we need a corresponding property
for the inverse metric $\hat h^{ij}$. 
\begin{lemma} \label{hhat}
Suppose that, in normal coordinates,  $\hat h_{ij}$ has the property that
$A_{ij}=\hat h_{ij}-\eta_{ij}$ satisfies $r^2 |\eta^{ij}A_{ij}$, $A_{ij}x^i=0$. Then the same holds for
the inverse $\hat h^{ij}$, i.e., with $B^{ij}=\hat h^{ij}-\eta^{ij}$, we have
$r^2 | \eta_{ij}B^{ij}$ and $B^{ik}\hat h_{jk}x^j=B^{ik} \eta_{jk} x^j=0$.
\end{lemma}
\begin{proof}
This is most easily seen in terms of matrices. Namely, it is easy to check that
the matrices $A_1,A_2,\ldots,A_7$ from Lemma \ref{gdec}
have the property that also the products $A_i A_j$, $1 \leq i,j \leq 7$
satisfies the assumptions of Lemma \ref{gdec}.
But this means that the the inverse 
$(I+A)^{-1}=I+B=I+\sum_{n=1}^\infty (-A)^n$ will be a sum of the identity operator $I$
and terms which all have the properties of \mbox{Lemma \ref{gdec}}.
\end{proof}

Let us now define the operator
$$D=x^j \frac{\partial}{\partial x^j}.$$
The operator $D$ has the important property that if $f(x,y,z)$ is a homogeneous
polynomial of order $n$, $D(f)=n D(f)$. This is true, whether $f$ is a scalar or tensor valued.
Because of this property, it easily follows that
\begin{lemma}
If $A_{ij}$ has the properties of Lemma \ref{gdec}, the same properties hold for $D(A_{ij})$.
\end{lemma}

Before we state and prove Lemma \ref{r2Rab}, we note that the definition of the Christoffel symbols ${\Gamma^k}_{ij}$ implies that in normal
coordinates,
$ 
2x^j {\Gamma^k}_{ij} =\hat h^{km} D( \hat h_{im}),
$
 and in particular 
$ 
2x^j {\Gamma^k}_{kj}=\hat h^{jk}D(\hat h_{jk})=D(\ln|\hat h|),
$
where $|\hat h|$ denotes the determinant of $\hat h_{ij}$.

\begin{lemma} \label{r2Rab}
Suppose that $(x^1,x^2,x^3)$ are normal coordinates on $(\hat V,\hat h_{ab})$. 
Then $r^2 | \eta^{ij}(\hat h_{ij}-\eta_{ij}) \Rightarrow r^2 | \mathbf{r}^{a}\mathbf{r}^{b}\hat R_{ab}$.
\end{lemma}
\begin{proof}
Since $(x^1,x^2,x^3)$ are normal coordinates, ${\Gamma^i}_{kl}x^k x^l=0$.
Using this ( and the fact that $\hat h_{ij}x^i=\eta_{ij} x^i$ ) the definition of $\hat R_{ij}$ gives
$$
x^i x^j R_{ij}=-x^i  \frac{\partial}{\partial x^i}(x^j { \Gamma^k}_{k j})
-x^j{ \Gamma^k}_{k j} -x^i x^j
 { \Gamma^m}_{k j}{ \Gamma^k}_{m i}.
 $$
Some further manipulation leads to
 \begin{equation} \label{r2rab}
 4x^i x^j R_{i j} 
  =  -2\hat h^{k m}D(\hat h_{k m})-D(\hat h^{k m})D(\hat h_{k m}) -2\hat h^{k m} D^2(\hat h_{k m}) 
\end{equation}
With the notation that $B$ stands for any matrix satisfying the properties of Lemma
\ref{gdec}, which means that $D^2(B)=D(B)=B$, $I\cdot B=B$, $B\cdot B=B$, $B+B=B$, 
(\ref{r2rab}) reads
\begin{equation*}\nonumber
4x^i x^j R_{i j} =[ -2 (I+B) B-B \cdot B-2 (I+B) B]=[B].
\end{equation*}
Since $[B]$ is divisible by $r^2$, so is $4x^i x^j R_{i j} $.
\end{proof}

We will now return to the original recursion (\ref{orgrec}). By combining the previous lemmas,
we have the following theorem, which allows for the direct connection between the moments
$ P_{a_1 \ldots a_n}$ and the potential $\hat \phi$.
\begin{theorem} \label{r2thm}
Let $(x^1,x^2,x^3)$ be normal coordinates on $\hat V$, $\hat h_{ij}-\eta_{ij}$ satisfy
the properties of Lemma \ref{gdec}, and let $ P_{a_1 \ldots a_n}$
 be defined by the recursion (\ref{orgrec}).
Then
$$x^{i_1} \ldots x^{i_n} P_{i_1 \ldots i_n} \equiv
x^{i_1} \ldots x^{i_n} \partial_{i_1}\cdots \partial_{i_n}P \pmod {r^2}$$
\end{theorem}
\begin{proof}
Put $c_n=\frac{n(2n-1)}{2}$. Then
\begin{equation*} \nonumber
\begin{array}{ccl}
x^{i_1} \ldots x^{i_n} P_{i_1 \ldots i_n} &
 = & x^{i_1} \ldots x^{i_n}C[\hat D_{i_1}P_{i_2 \ldots i_n}-
c_{n-1}\hat R_{i_1 i_2}P_{i_3 \ldots i_n}]\\
& = & x^{i_1} \ldots x^{i_n}
C[\hat D_{i_1}P_{i_2 \ldots i_n}]-
c_{n-1}x^{i_1} \ldots x^{i_n}C[\hat R_{i_1 i_2}P_{i_3 \ldots i_n}]\\
& \equiv & x^{i_1} \ldots x^{i_n}
\hat D_{i_1}P_{i_2 \ldots i_n}-
c_{n-1}x^{i_1} \ldots x^{i_n}\hat R_{i_1 i_2}P_{i_3 \ldots i_n} \\
& \equiv & x^{i_1} \ldots x^{i_n}
\hat D_{i_1}P_{i_2 \ldots i_n} \pmod {r^2}\\
\end{array}
\end{equation*}
Now, 
\begin{equation*} \nonumber
x^{i_1} \ldots x^{i_n}\hat D_{i_1}P_{i_2 \ldots i_n}=x^{i_1} \ldots x^{i_n}\hat D_{i_1}C[\hat D_{i_2}P_{i_3 \ldots i_n}-
c_{n-2} \hat R_{i_2 i_3}P_{i_4 \ldots i_n}]\\
\end{equation*}
and for some tensor $S_{i_4 \ldots i_n}$
\begin{equation*} \nonumber
\hspace*{-10mm}\begin{array}{rl}
x^{i_1} \ldots x^{i_n}\hat D_{i_1}
C[\hat D_{i_2}P_{i_3 \ldots i_n}]&=
x^{i_1} \ldots x^{i_n}\hat D_{i_1}(\hat D_{(i_2}P_{i_3 \ldots i_n)}-\hat h_{(i_2 i_3}S_{i_4 \ldots i_n)})\\
&=x^{i_1} \ldots x^{i_n}\hat D_{i_1}\hat D_{i_2}P_{i_3 \ldots i_n}-
x^{i_1} \ldots x^{i_n}\hat h_{i_2 i_3}\hat D_{i_1}S_{i_4 \ldots i_n}\\
&\equiv x^{i_1} \ldots x^{i_n}\hat D_{i_1}\hat D_{i_2}P_{i_3 \ldots i_n}   \pmod {r^2},
\end{array}
\end{equation*}
while
\begin{equation*} \nonumber
\begin{array}{r}
x^{i_1} \ldots x^{i_n}\hat D_{i_1}C[\hat R_{i_2 i_3}P_{i_4 \ldots i_n}]
\equiv x^{i_1} \ldots x^{i_n}\hat D_{i_1}(\hat R_{i_2 i_3}P_{i_4 \ldots i_n}) \\
\equiv x^{i_1} \ldots x^{i_n}(\hat D_{i_1}\hat R_{i_2 i_3}) P_{i_4 \ldots i_n} \equiv 0  \pmod {r^2}\\
\end{array}
\end{equation*}
where, in the last step, we have used that for some function $f$, $x^i x^j x^k \hat D_i \hat  R_{jk}=
x^i \hat D_i (x^j x^k \hat  R_{jk})-\hat  R_{jk} x^i \hat D_i (x^j x^k)=
D(r^2 f)-\hat  R_{jk} x^i \partial_i (x^j x^k)
+ R_{jk} x^i {\Gamma^j}_{im}x^m x^k+ R_{jk} x^i {\Gamma^k}_{im}x^j x^m=
2r^2 D(f)+r^2 D(f)-2\hat  R_{jk}x^j x^k \equiv 0   \pmod {r^2}$.
Thus,
$$
x^{i_1} \ldots x^{i_n} P_{i_1 \ldots i_n} \equiv 
x^{i_1} \ldots x^{i_n}\hat D_{i_1}\hat D_{i_2}P_{i_3 \ldots i_n}   \pmod {r^2}.
$$
Proceeding in the same way, we find that
$$x^{i_1} \ldots x^{i_n} P_{i_1 \ldots i_n} \equiv
x^{i_1} \ldots x^{i_n} \hat D_{i_1}\cdots \hat D_{i_n}P \pmod {r^2}.$$
Moreover,
\begin{equation*} \nonumber
\hspace*{-10mm}\begin{array}{l}
x^{i_1} \ldots x^{i_n}\hat D_{i_1}\ldots \hat D_{i_n}P  \\
=
x^{i_1} \ldots x^{i_n}\partial_{i_1}\hat D_{i_2}\ldots \hat D_{i_n}P  -
x^{i_1} \ldots x^{i_n}\sum_{m=2}^n {\Gamma^m}_{i_1 i_m}
\hat D_{i_2}\ldots \hat D_{m} \cdots \hat D_{i_n}P\\
=
x^{i_1} \ldots x^{i_n}\partial_{i_1}\hat D_{i_2}\ldots \hat D_{i_n}P\\
=
x^{i_1} \ldots x^{i_n}\partial_{i_1}\partial_{i_2}\hat D_{i_3}\ldots \hat D_{i_n}P  -
x^{i_1} \ldots x^{i_n}\partial_{i_1}\sum_{m=3}^n {\Gamma^m}_{i_1 i_m}
\hat D_{i_3}\ldots \hat D_{m} \ldots \hat D_{i_n}P\\
=
x^{i_1} \ldots x^{i_n}\partial_{i_1}\partial_{i_2}\hat D_{i_3}\cdots \hat D_{i_n}P 
\end{array}
\end{equation*}
where, in the last step, we have used 
$x^i x^j x^k \partial_i {\Gamma^m}_{jk}=x^i \partial_i ( x^j x^k{\Gamma^m}_{jk})$\newline $-
x^i \partial_i ( x^j x^k){\Gamma^m}_{jk}=0$.
Again we can proceed and get the statement of the theorem.
\end{proof}
Remark. 
This theorem is comparable to Theorem 7 in \cite{bahe}. The difference lies in the presentation
since \cite{bahe} uses complex vectors, and the effect is that the equivalence $\pmod {r^2}$ here becomes
equality in \cite{bahe} due to the fact that $r^2=0$ along the null vectors there. Also, although
\cite{bahe} uses null vectors, they are  'complexified unit vectors', while the statement here
uses vectors $x^i$ which are not normalized. However, in the statement of Theorem \ref{r2thm},
one can (as commented before) replace each $x^k$ by the direction dependent unit vector
$\rho^k=x^k/r$. 

Comparing with the definition of $P=\hat \phi=\sum_{|\alpha|\geq 0}\frac{\mathbf{r}^\alpha}{\alpha!}
P^0_\alpha $ in the beginning of the section, Taylor's theorem together with Theorem \ref{r2thm}
and the remark after Lemma \ref{normal}
tells us that the multipoles produced will be precisely the desired multipoles $P^0_\alpha$,
and we may note that by the arguments presented, the multipoles will be unaffected by a change
$\hat \phi \to \hat \phi +r^2 \gamma$.

Thus, by specifying $\hat \phi$, and by requiring that the metric $\hat h_{ij}$ be as in
Theorem \nolinebreak \ref{r2thm}, the recursion (\ref{orgrec}) produces the prescribed multipoles
moments. The issue is now whether (\ref{confeq}) produces a 
power series for such a $\hat h_{ij}$, and furthermore if this series converges.

\subsection{The conformal field equations} \label{confeqsec}
In this section, we will address the conformal field equations (\ref{confeq}). 
The requirement that these equations extend smoothly to $i^0$ will put further restrictions
on the form of the metric, and also involve the conformal factor $\Omega$.

We consider $\hat \phi$ as fixed and real analytic with respect to the normal coordinates
$(x,y,z)$ in a neighbourhood of $i^0=(0,0,0)$. 
From Section \ref{tmoments}, we have that ($\omega=0$)
$$
\phi=\frac{\lambda^2-1}{4\lambda}, \quad P= \hat \phi=\phi/\sqrt \Omega,
$$
where $\lambda$, which is the norm of the Killing vector, and the conformal factor $\Omega$ appear in 
(\ref{confeq}). The conformal factor $\Omega$ which must satisfy the conditions in Section
\ref{tmoments}, is not unique. Rather, we have the freedom $\Omega \to \Omega e^{\kappa}$,
where $\kappa$ is a formal power series which vanish at $i^0$. Also, 
$\hat D_a \kappa (0)$ is known to mix the moments, corresponding to a 'translation' in the classical sense. It is therefore natural to demand $\hat D_a \kappa (0)=0$.
Solving for $\lambda$, we thus get
$$\lambda=\sqrt{1+4 \Omega \hat \phi^2}+2 \sqrt{\Omega}\hat \phi,\quad
\Omega=r^2 e^{2 \kappa}, \kappa(0)=\hat D_a \kappa(0)=0.
$$
It is important
to note that $\lambda$ is not even formally smooth, i.e., we cannot regard $\lambda$
as a formal power series. This is of course due to the occurrence of $\sqrt{\Omega}=
r e^{\kappa}$, where $r=\sqrt{x^2+y^2+z^2}$ is non-regular at $i^0$. In effect, this
will mean that each of the equations in (\ref{confeq}) will split into two, seemingly doubling 
the number of equations.

To address (\ref{confeq}) we split 
$\ln (\sqrt{\lambda} \Omega)=\frac{1}{2}\ln \lambda+\ln(r^2 e^{2 \kappa})$ and note the
convenient  relation
$$
\frac{1}{2}\hat D_a \ln \lambda=\frac{\hat D_a(r \hat \phi e^\kappa)}{\sqrt{1+4 r^2 \hat \phi^2 e^{2\kappa}}}
$$
Inserted in (\ref{confeq}), this gives the equations
\begin{equation} \label{rab1}
\begin{array}{l}
\hat R_{ab}+\hat D_a \hat D_b \ln(r^2 e^{2 \kappa})+\hat D_a \ln(r^2 e^{2 \kappa})\hat D_b \ln(r^2 e^{2 \kappa})\\
+\hat h_{ab} \hat h^{de} \hat D_d \hat D_e \ln(r^2 e^{2 \kappa})
-\hat h_{ab}\hat h^{de}\hat D_d \ln(r^2 e^{2 \kappa})\hat D_e\ln(r^2 e^{2 \kappa})\\
+\hat h_{ab} \hat h^{de} \hat D_d \frac{\hat D_e(r \hat \phi e^\kappa)}{\sqrt{1+4 r^2 \hat \phi^2 e^{2\kappa}}}
 -\hat h_{ab}\hat h^{de}\frac{\hat D_d(r \hat \phi e^\kappa)\hat D_e\ln(r^2 e^{2 \kappa})}{\sqrt{1+4 r^2 \hat \phi^2 e^{2\kappa}}}\\
-2 \frac{\hat D_a(r \hat \phi e^\kappa)\hat D_b(r \hat \phi e^\kappa)}{1+4 r^2 \hat \phi^2 e^{2\kappa}}=0.
\end{array}
\end{equation}
and
\begin{equation} \label{r1}
\begin{array}{l}
\hat R+4 \hat h^{ab}\hat D_a \hat D_b \ln(r^2 e^{2 \kappa})
+4 \hat h^{ab}\hat D_a \frac{\hat D_b (r \hat \phi e^\kappa)}{\sqrt{1+4 r^2 \hat \phi^2 e^{2\kappa}}}\\
-2 \hat h^{ab}\frac{\hat D_a(r \hat \phi e^\kappa)\hat D_b(r \hat \phi e^\kappa)}{1+4 r^2 \hat \phi^2 e^{2\kappa}}
-4 \hat h^{ab}\frac{\hat D_a(r \hat \phi e^\kappa)\hat D_b \ln(r^2 e^{2 \kappa})}{\sqrt{1+4 r^2 \hat \phi^2 e^{2\kappa}}}
\\
-2 \hat h^{ab}\hat D_a \ln(r^2 e^{2 \kappa})\hat D_b \ln(r^2 e^{2 \kappa})=0.
\end{array}
\end{equation}
To continue, we need some relations which hold in our specialized coordinate system.
Each of the following statements are straightforward to check.
\begin{equation} \label{prop} \hspace*{-10mm}
\begin{array}{l}
\hat D_i \hat D_j r^2=2\hat h_{ij}+D(\hat h_{ij}), 
 \qquad
\Delta r^2=\hat h^{ij}\hat D_i \hat D_j r^2=6+\hat h^{ij} D(\hat h_{ij}),\ \\
\hat h^{ij} D(\hat h_{ij})=D(\ln |\hat h|) \equiv 0 \pmod {r^2}, 
  \qquad
\forall f: \hat h^{jij} \hat D_i r^2 \hat D_j f = 2 D(f), \\
\hat h^{ij} \hat D_i r \hat D_j r =1.
\end{array}
\end{equation}
We will now split (\ref{rab1}) and (\ref{r1}) into their regular and non-regular parts.
By a regular function we mean a function $f=f(x,y,z)$ such that $r^{2n} f$ is (formally)
real analytic for some integer $n \geq 0$. Note that if $n \geq 1$ is required, the regular
function is singular. Similarly, a function $f$ is non-regular if $r^{2n-1} f$ is (formally)
real analytic for some integer $n \geq 0$.  Again, if $n \geq 1$ is required, the function is non-regular
and singular. This division is due to the non-regularity of $\lambda$, and all functions
or tensor fields can be written as $f =f_1+r f_2$ where $f_1$ and $f_2$ are regular. (Cf. Lemma 2
of \cite{bahe}.)  In particular, $f=0$ requires $f_1=f_2=0$.
We note that $\hat D_a r=\frac{1}{2r}\hat D_a r^2$, which is non-regular.
On the other hand,
$\hat D_a(r \hat \phi e^\kappa)\hat D_b(r \hat \phi e^\kappa)=
[\frac{1}{2r}\hat D_a ( r^2 )\hat \phi e^\kappa+r \hat D_a(\hat \phi e^\kappa)]
[\frac{1}{2r}\hat D_b (r^2 )\hat \phi e^\kappa+r \hat D_b(\hat \phi e^\kappa)]=
\frac{\hat \phi^2 e^{ 2\kappa}}{4r^2}\hat D_a  r^2 \hat D_b  r^2 +
r^2 \hat D_a(\hat \phi e^\kappa) \hat D_b(\hat \phi e^\kappa)+
 \frac{1}{2}\hat D_{(a}(\hat \phi e^\kappa)\hat D_{b)} ( r^2 )\hat \phi e^\kappa 
$
which is regular (but singular).
Therefore, the regular part of (\ref{rab1}) is
\begin{equation} \label{rab2a}
\begin{array}{l}
\hat R_{ab}+\hat D_a \hat D_b \ln(r^2 e^{2 \kappa})+\hat D_a \ln(r^2 e^{2 \kappa})\hat D_b \ln(r^2 e^{2 \kappa})\\
+\hat h_{ab} \hat h^{de} \hat D_d \hat D_e \ln(r^2 e^{2 \kappa})
-\hat h_{ab}\hat h^{de}\hat D_d \ln(r^2 e^{2 \kappa})\hat D_e\ln(r^2 e^{2 \kappa})\\
-2 \frac{\hat D_a(r \hat \phi e^\kappa)\hat D_b(r \hat \phi e^\kappa)}{1+4 r^2 \hat \phi^2 e^{2\kappa}}=0.
\end{array}
\end{equation}
while the non-regular part gives
\begin{equation} \label{rab2b}
\hat h_{ab} \hat h^{de} \hat D_d \frac{\hat D_e(r \hat \phi e^\kappa)}{\sqrt{1+4 r^2 \hat \phi^2 e^{2\kappa}}}
 -\hat h_{ab}\hat h^{de}\frac{\hat D_d(r \hat \phi e^\kappa)\hat D_e\ln(r^2 e^{2 \kappa})}{\sqrt{1+4 r^2 \hat \phi^2 e^{2\kappa}}}=0.
\end{equation}
Similarly, (\ref{r1}) splits into the two equations
\begin{equation} \label{r2a}
\begin{array}{l}
\hat R+4 \hat h^{ab}\hat D_a \hat D_b \ln(r^2 e^{2 \kappa})
-2 \hat h^{ab}\frac{\hat D_a(r \hat \phi e^\kappa)\hat D_b(r \hat \phi e^\kappa)}{1+4 r^2 \hat \phi^2 e^{2\kappa}}\\
-2 \hat h^{ab}\hat D_a \ln(r^2 e^{2 \kappa})\hat D_b \ln(r^2 e^{2 \kappa})=0
\end{array}
\end{equation}
and (dividing by 4)
\begin{equation} \label{r2b}
\hat h^{ab}\hat D_a \frac{\hat D_b (r \hat \phi e^\kappa)}{\sqrt{1+4 r^2 \hat \phi^2 e^{2\kappa}}}
-\hat h^{ab}\frac{\hat D_a(r \hat \phi e^\kappa)\hat D_b \ln(r^2 e^{2 \kappa})}{\sqrt{1+4 r^2 \hat \phi^2 e^{2\kappa}}}
=0.
\end{equation}
Equations (\ref{rab2b}) and (\ref{r2a}) are redundant, since (\ref{rab2b}) is a multiple of 
(\ref{r2b}), while (\ref{r2a}) is the trace of (\ref{rab2a}). Since (\ref{r2b}) is non-regular,
$r \cdot \!\!$ (\ref{r2b}) is regular, and by using 
$r \hat D_i \hat D_j r=\frac{1}{2}\hat D_i \hat D_j r^2-\frac{1}{4r^2} D_i r^2 D_j r^2$,
it is also seen that $r \cdot \!\!$ (\ref{r2b})  is smooth at $i^0$. We therefore look at the singular 
part\footnote{In practise, skipping smooth terms} of
(\ref{rab2a}), which is found to be
\begin{equation} \label{smooth1}
\begin{array}{l}
\frac{1}{r^2}
\hat D_i \hat D_j r^2+\frac{2}{r^2}\hat D_i r^2 \hat D_j \kappa+
\frac{2}{r^2}\hat D_i \kappa \hat D_j r^2\\
+\hat h_{ij}[\frac{1}{r^2}\Delta r^2-\frac{8}{r^2}-
\frac{8}{r^2}D(\kappa)]-\frac{\hat \phi^2 e^{2\kappa}}{2r^2}\hat D_i r^2 \hat D_j r^2. 
\end{array}
\end{equation}
This expression must be smooth, and by taking its trace, this says that
$
\frac{4}{r^2}\Delta r^2-\frac{24}{r^2}-\frac{16}{r^2} D(\kappa)
$
must be smooth. However, from (\ref{prop}), this implies that
$\frac{D(\kappa)}{r^2}$ is smooth, and therefore that $\kappa=C+r^2\chi$,
where $C$ is a constant, and $\chi$ is smooth. From $\kappa(i^0)=0$ we thus infer that
\begin{equation} \label{chi}
\kappa=r^2 \chi
\end{equation}
for some smooth function $\chi$. Inserting (\ref{chi}) in (\ref{smooth1}),
some simplification shows that
\begin{equation} \label{smooth2}
\frac{1}{r^2}[2 D(\hat h_{ij})+\hat D_i r^2 \hat D_j r^2 (8 \chi -\hat \phi^2)]
\end{equation}
must be smooth. This equation, as it stands, has many solutions, and we must choose a solution
which still allows the metric $\hat h_{ij}$ and the function $\kappa=r^2 \chi$ to solve the equations
(\ref{rab2a}) and (\ref{r2b}).  To do this, we start by imposing slightly more conditions on
$\hat h_{ij}$, so that $\hat h_{ij}$ will still satisfy the conditions of Lemma
\ref{hhat}, but in a slightly restricted form. Let $A_1$ be the matrix from
Lemma \ref{gdec}. We will then require that $\hat h_{ij}$ takes the form
\begin{equation}\label{met2}
\hat h_{ij}=\eta_{ij}+f (x,y,z)(A_1)_{ij}+r^2 \gamma_{ij}
\end{equation}
Since we have the factor $r^2$ explicitly in front of $\gamma_{ij}$, we only need to ensure that $\gamma_{ij} x^i =0$
in order for $\hat h_{ij}$ to still satisfy the conditions of Lemma \ref{hhat}.  This is guaranteed by the following lemma.

\begin{lemma} \label{nygdeco}
Suppose that $\gamma_{ij}$ is such that  $\gamma_{ij} x^i =0$. Then $\gamma_{ij}$ is uniquely
decomposable as 
$$
\begin{array}{rcl}
\gamma_{ij}=\gamma &=&f_1(x,y) B_1+f_2(x,y) B_2+f_3(x,y) B_3+f_4(x,y,z) B_4\\
&&+f_5(x,y,z) B_5+f_6(x,y,z) B_6, \mbox{ where}
\end{array}
$$
$
\begin{array}{rcl}
B_1&=&\left(
\begin{array}{lll}
 -y^2 & x y & 0 \\
 x y & -x^2 &  \\
 0 & 0 & 0
\end{array}
\right)\\
B_2&=&\left(
\begin{array}{lll}
2 y z & -x z & -x y \\
 -x z  & 0 & x^2  \\
 -x y & x^2  & 0
\end{array}
\right)\\
B_3&=&\left(
\begin{array}{lll}
 0 & -y z & y^2 \\
- y z & 2 x z & -x y  \\
y^2 & -x y  & 0
\end{array}
\right)\\
B_4&=&\left(
\begin{array}{lll}
 z^2 & 0 & -x z \\
 0 & 0 & 0 \\
-x z & 0 & x^2
\end{array}
\right)\\
B_5&=&\left(
\begin{array}{lll}
 0 & z^2 & -y z \\
 z^2 & 0 & -x z \\
 -y z & -x z & 2 x y
\end{array}
\right)\\
B_6&=&\left(
\begin{array}{lll}
0 & 0 & 0 \\
 0 & z^2 & -y z \\
0 & -y z & y^2
\end{array}
\right)\\
\end{array}
$
\end{lemma}
\begin{proof}
This proof is rather similar to the proof of Lemma \ref{gdec} and is given in appendix A.
\end{proof}
Now, $(A_1)_{ij}=x_i x_j -r^2 \eta_{ij}$, and $ 4 D(f x_i x_j)=(2f+ D(f))\hat D_i r^2 \hat D_j r^2$. Using this, insertion
of (\ref{met2}) into (\ref{smooth2}) then leaves us with the condition that 
\begin{equation} \label{smooth3}
\frac{\hat D_i r^2 \hat D_j r^2}{r^2}[2f+ D(f)+ 2 (8 \chi -\hat \phi^2)]
\end{equation}
must be smooth at $r=0$. 

It is not trivial to impose the right conditions on $f$ and $\chi$. If they are chosen too restrictively,
no solution to (\ref{rab2a}) and (\ref{r2b}) will exist. On the other hand, if $f$ and $\chi$ are not restricted enough,
the solution we are looking for will not be unique (in terms of the introduced quantities). As we will see, the following
choice, upon which we insist, will suffice.
\begin{equation} \label{chidef}
\chi=\frac{1}{16}(2\hat \phi^2-2 f-  D(f))
\end{equation}
The choice (\ref{chidef}) will make (\ref{smooth3}) vanish identically; this means that the equations
(\ref{rab2a}) and (\ref{r2b}) now are smooth at $i^0$. Also, the form of the metric is determined
via Lemma \ref{nygdeco}, with $f(x,y,z), f_1(x,y), f_2(x,y), f_3(x,y), f_4(x,y,z), f_5(x,y,z)$ and $f_6(x,y,z)$ as unknowns. We put $f_0=f$ and note that
with $\hat h_{ij}$ known the conformal factor $\Omega$, and thus the desired spacetime, is also determined.
\subsection{Uniqueness of the metric} \label{unique}
We will now address the equations (\ref{rab2a}) and (\ref{r2b}).
We will first prove that a part of these equations determine the metric uniquely provided
$\hat h_{ij}$ is cast in a certain way (Lemma \ref{f1}), and then, in Section \ref{conv}, that the derived series expansions
for $\hat h_{ij}$ converges and also satisfies  (\ref{rab2a}) and (\ref{r2b}) in full (Lemma \ref{f2}).

To prepare for Lemma \ref{f1} and Lemma \ref{f2}, we will create
some scalar equations from (\ref{rab2a}). To do this, we introduce the vectors fields
$u_1^a, u_2^a, u_3^a$ with components
\begin{equation*} \nonumber
u_1^i=
\left(
\begin{array}{c}
  -y   \\
  x   \\
  0   
\end{array}
\right), \quad
u_2^i=
\left(
\begin{array}{c}
  0  \\
  -z   \\
  y   
\end{array}
\right), \quad
u_3^i=
\left(
\begin{array}{c}
  z   \\
  0   \\
  -x   
\end{array}
\right)
\end{equation*}
where it is seen that they are all pointwise orthogonal to the vector field $x^i$. 
We also note the linear relation
$$
z u_1^i + x u_2^i+y u_3^i=0.
$$
It is straightforward to check that
$u_1^i u_1^j, u_2^i u_2^j, u_3^i u_3^j, x^i x^j$ together with any two
of $u_1^{(i}x_{\phantom{1}}^{j)}, u_2^{(i}x_{\phantom{1}}^{j)}, 
u_3^{(i}x_{\phantom{1}}^{j)}$ are (point-wise) linearely independent
Next, denote the LHS of (\ref{rab2a})
by $T_{ab}$, and denote $r\sqrt{1+4 r^2 \hat \phi^2 e^{2\kappa}}$ times the LHS of (\ref{r2b}) by $S$.
With the notation 
$$
t_{11}=u_1^i u_1^j T_{ij}, \quad
t_{22}=u_2^i u_2^j T_{ij}, \quad
t_{33}=u_3^i u_3^j T_{ij}
$$
$$
t_{00}=x^i x^j T_{ij},\quad t_{0k}=x^i u_k^j T_{ij}, k=1,2,3,
$$
it is clear that (\ref{rab2a}) is satisfied if and only if $t_{00}, t_{11}, t_{22}, t_{33}$
and any two of $t_{0k}, k=1,2,3$ vanishes.

\begin{lemma} \label{f1} 
Suppose that the metric components $\hat h_{ij}$ takes the form
\begin{equation}
\begin{array}{r}
\hat h_{ij}=\hat h=\eta+f_0(x,y,z)A_1+ r^2 \left[ f_1(x,y) B_1+f_2(x,y) B_2+f_3(x,y) B_3  \right. \\
\left. +f_4(x,y,z) B_4
+f_5(x,y,z) B_5+f_6(x,y,z) B_6 \right],
\label{metricform}
\end{array}
\end{equation}
 where $A_1$ and $B_1, B_2, \cdots B_6$
 are the matrices
 in Lemmas \ref{gdec} and \ref{nygdeco} respectively, and where the functions
$f_0, f_1, \cdots f_6$ are formal analytical functions of the variables indicated.
Then the equations
$$
S=0,\quad  t_{11}=0, \quad t_{22}=0, \quad t_{33}=0
$$
determines the metric $\hat h_{ij}$ as formal power series.
\end{lemma}
This will be proved by induction. With the notation from Section \ref{notation},
$$
\hat h_{ij}=\eta_{ij}+\sum_{n=2}^\infty \hat h_{ij}^{[n]},
\quad \hat R_{ij}=\sum_{n=0}^\infty \hat R_{ij}^{[n]},
 \quad T_{ij}=\sum_{n=0}^\infty T_{ij}^{[n]},
\quad S=\sum_{n=0}^\infty S^{[n]}.
$$
Note that there can be no linear term in the metric corresponding to 
$n=1$ due to Lemma \ref{gdec}.
Moreover, $\hat h_{ij}$ is determined by the functions $f_0, \cdots f_6$ with corresponding
series
\begin{equation*} \nonumber
\begin{array}{rcl}
f_k(x,y,z)&=&\sum_{n=0}^\infty f_k^{[n]}(x,y,z), \ k=0,4,5,6,\\
f_k(x,y)&=&\sum_{n=0}^\infty f_k^{[n]}(x,y), \ k=1,2,3,
\end{array}
\end{equation*}
where each $f_k^{[n]}$ is a homogeneous polynomial (in the variables indicated) of degree $n$.
We also recollect that $\kappa$ is determined by $f_0$ and $\hat \phi$ via (\ref{chi}) and
(\ref{chidef}).

$\hat R_{ij}^{[n]}$ is determined by the metric up to order $n+2$, and
with $\hat h_{ij}^{[n+2]}$ as the leading term, while
$\hat h_{ij}^{[k]}, k \leq n+1$ are regarded  as lower order terms (L.O.T.) , the definition of the Ricci
tensor gives that 
\begin{equation} \label{riccilead}
\begin{array}{rcl}
\hat R_{ij}^{[n]}&=&
 \hat{\underline R}_{ij}^{[n]} +\mbox{L.O.T.}\\
\hat{\underline R}_{ij}^{[n]} &=&
\frac{1}{2} \eta^{m n} \partial_m\left\{ \partial_j h_{i n}^{[n+2]} 
+\partial_i  \hat h_{j n}^{[n+2]}-\partial_n  \hat h_{i j}^{[n+2]}
\right\}\\
&-&
\frac{1}{2} \eta^{m n}\partial_i\left\{\partial_m \hat h_{j n}^{[n+2]}
+\partial_j \hat h_{m n}^{[n+2]} -\partial_n\hat h_{m j}^{[n+2]}
\right\}
\end{array}
\end{equation}
We stress that the lower order terms in $\hat R_{ij}^{[n]}$ are also polynomials of degree $n$, 
but that they are expressions in  $\hat h_{ij}^{[k]}, k \leq n+1$.
Considering the form of $\hat h_{ij}$, 
$ \hat{\underline R}_{ij}^{[n]}$ is a function of the leading order polynomials $f_0^{[n]}$, $f_k^{[n-2]}$ for $1\leq k \leq 6$. We now proceed and look at
$T_{ij}$ and see how the leading order polynomials enter. 
From (\ref{rab2a}) we find 
\begin{equation} \label{alsolead}
\begin{array}{l}
(\hat D_i \hat D_j \ln(r^2))^{[n]}
=
(\frac{1}{r^2}\hat D_i \hat D_j r^2-\frac{1}{r^4}\hat D_i r^2 \hat D_j r^2)^{[n]}\\
=
\frac{1}{r^2}(
2\hat h_{ij}^{[n+2]}+D(\hat h_{ij}^{[n+2]})
)+\mbox{L.O.T.}=\frac{n+4}{r^2}\hat h_{ij}^{[n+2]}+\mbox{L.O.T.}
\end{array}
\end{equation}
Here we have used that $D(f)=n f$ if $f$ is homogeneous of order $n$. Continuing, (\ref{rab2a})
gives
\begin{equation} \label{tablead} \hspace*{-15mm}
\begin{array}{rl}
2 (\hat D_i \hat D_j \kappa)^{[n]}& =
(\frac{1}{8}\hat D_i \hat D_j (r^2 (2\hat \phi^2- 2 f -
 D(f))))^{[n]} \\
 &=
\frac{-(2+n)}{8}\partial_{i} \partial_{j} (r^2  f_0^{[n]})+\mbox{L.O.T.}
\\
(\hat D_i \ln(r^2 e^{2 \kappa})\hat D_j \ln(r^2 e^{2 \kappa}))^{[n]}&=
\frac{-(n+2)}{4r^2}\partial_{(i} r^2 \partial_{j)}(r^2 f_0^{[n]})+\mbox{L.O.T.}
\\
(\hat h_{ij} \hat h^{nm} \hat D_n \hat D_m \ln(r^2))^{[n]}&=
\frac{n+2}{r^2}\eta_{ij}\eta^{nm}\hat h_{nm}^{[n+2]}+\frac{2}{r^2}\hat h_{nm}^{[n+2]}+\mbox{L.O.T.}
\\
(2\hat h_{ij} \hat h^{nm} \hat D_n \hat D_m  \kappa)^{[n]}&=
\frac{-(2+n)}{8}\eta_{ij}\eta^{nm} \partial_n \partial_m (r^2  f_0^{[n]})+\mbox{L.O.T.}
\\
(\hat h_{ij}\hat h^{nm}\hat D_n \ln(r^2 e^{2 \kappa})\hat D_m\ln(r^2 e^{2 \kappa}))^{[n]}&=
-\eta_{ij}\frac{(n+2)^2}{2}f_0^{[n]}+\frac{4}{r^2}\hat h_{nm}^{[n+2]}+\mbox{L.O.T.}
\\
\left(2 \frac{\hat D_i(r \hat \phi e^\kappa)\hat D_j(r \hat \phi e^\kappa)}{1+4 r^2 \hat \phi^2 e^{2\kappa}}\right)
^{[n]}&=
\ 0+\mbox{L.O.T}.
\end{array}
\end{equation}
Combining (\ref{riccilead}), (\ref{alsolead}) and (\ref{tablead}), we find the leading order of
$T_{ij}^{[n]}=\underline T_{ij}^{[n]}$+ L.O.T. to be
\begin{equation*}  \nonumber
\begin{array}{cl}
\underline T_{ij}^{[n]}&=
\frac{1}{2} \eta^{m n} \partial_m\left\{ \partial_j \hat h_{i n}^{[n+2]} 
-\partial_n  \hat h_{i j}^{[n+2]}
\right\}
-\frac{1}{2} \eta^{m n}\partial_i\left\{
\partial_j \hat h_{m n}^{[n+2]} -\partial_n\hat h_{m j}^{[n+2]}
\right\}\\
&+\frac{n+2}{r^2}\hat h_{ij}^{[n+2]}+\frac{n+2}{r^2}\eta_{ij}\eta^{nm}\hat h_{nm}^{[n+2]}+
\eta_{ij}\frac{(n+2)^2}{2}f_0^{[n]}\\
&-\frac{(n+2)}{8}[\partial_i \partial_j (r^2  f_0^{[n]})
+\frac{2}{r^2}\partial_{(i} r^2 \partial_{j)}(r^2 f_0^{[n]})
+\eta_{ij}\eta^{nm} \partial_n \partial_m (r^2  f_0^{[n]})]
\end{array}
\end{equation*}
To examine $S$, we write $Y=2\hat \phi e^\kappa$, which means that $S/\sqrt{1+4 r^2 \hat \phi^2 e^{2\kappa}}$ becomes
\begin{equation*} \nonumber 
\begin{array}{l}
\hat h^{ij}r\hat D_i \frac{\hat D_j (r Y)}{\sqrt{1+r^2 Y^2}}
-\hat h^{ij}\frac{r\hat D_i(r Y)\hat D_j \ln(r^2 e^{2 \kappa})}{\sqrt{1+r^2 Y^2}}\\
=-\hat h^{ij} \frac{r^2 Y \hat D_i(r Y) \hat D_j(r Y)}{\sqrt{1+r^2 Y^2}^3}+
\hat h^{ij}\frac{r\hat D_i \hat D_j (r Y)}{\sqrt{1+r^2 Y^2}}
-\hat h^{ij}\frac{r\hat D_i(r Y)[\hat D_j \ln(r^2)+2 \hat D_j \kappa]}{\sqrt{1+r^2 Y^2}}\\
\end{array}
\end{equation*}
Expanding the derivatives and using (\ref{prop}), 
we find
\begin{equation} \label{S2}
\begin{array}{cl}
S&=\frac{Y}{2} D(\ln|\hat h|)+r^2\Delta Y-2 YD(\kappa)-2r^2\hat h^{ij}\hat D_i Y \hat D_j \kappa\\
&-\frac{r^2 Y}{1+r^2 Y^2}[Y^2 +2 Y D(Y)+r^2 \hat h^{ij}\hat D_i Y \hat D_j]
\end{array}
\end{equation}
Note that $S \equiv 0 \pmod {r^2}$.
Now, using that $(\ln |\hat h|)^{[n+2]}=[\hat h_{ij}^{[n+2]}]+\mbox{L.O.T.}$, and putting
$\Delta_C=\eta^{ij}\partial_i \partial_j$, 
it follows that (at this stage,  $\hat \phi(0)\neq 0$ by assumption) 
$S^{[n+2]}=\underline S^{[n+2]}$+L.O.T., where
\begin{equation} \label{slead}
\begin{array}{cl}
\underline S^{[n+2]}&=
\frac{\hat \phi(0)}{2}((n+2)[\hat h_{ij}^{[n+2]}]+2 r^2 \Delta_c \kappa^{[n+2]}
-4 D(\kappa^{[n+2]}))\\
&=\frac{\hat \phi(0)(n+2)}{2}([\hat h_{ij}^{[n+2]}]-\frac{1}{8} r^2 \Delta_c (r^2 f_0^{[n]}
)+\frac{1}{4}D(r^2 f_0^{[n]})).
\end{array}
\end{equation}

The leading order terms in (\ref{riccilead}), (\ref{tablead}), (\ref{slead})
are all functions of the homogeneous polynomials $f_0^{[n]}$ and $f_k^{[n-2]}, k=1,2,\cdots 6$.
$f_0^{[n]}(x,y,z)$ defines  a vector space with dimension $\frac{(n+2)(n+1)}{2}$, while
$f_k^{[n]}(x,y)$ ($k=1,2,3$) define a vector space with dimension $n+1$.
In total, the tuple (for each fixed $n$)
\begin{equation*} \nonumber 
\overline p=(f_0^{[n]},f_1^{[n-2]},f_2^{[n-2]},f_3^{[n-2]},
f_4^{[n-2]},f_5^{[n-2]},f_6^{[n-2]})
\end{equation*}
 lives in a  vector space with dimension
$2 n^2 +3n-2$ for $n \geq 2$.
On the other hand, $T_{ij}^{[n]}$ will show up in $t_{11}^{[n+2]}, t_{22}^{[n+2]},
t_{33}^{[n+2]}$, which are all polynomials of degree $\frac{(n+4)(n+3)}{2}$, so that
 the quadruple 
\begin{equation*} \nonumber 
\overline t= (S^{[n+2]},t_{11}^{[n+2]},t_{22}^{[n+2]},t_{33}^{[n+2]})
\end{equation*}
lives in $\mathbf{R}^{2 n^2+14 n+24}$.

Thus, for each fixed $n$, $\overline t$ is a function of $\overline p$ and the functions
$f_i^{[\cdot]}$ of lower order (than in $\overline p$).
One also notes that both $(S^{[0]},t_{11}^{[0]},t_{22}^{[0]},t_{33}^{[0]})$ and 
$(S^{[1]},t_{11}^{[1]},t_{22}^{[1]},t_{33}^{[1]})$ are identically zero (see for instance Lemma 
\ref{reddim1} below). 
Next, it is easy to explicitly check that for $n=0$, one can find $f_0^{[0]}$ so that
$(S^{[2]},t_{11}^{[2]},t_{22}^{[2]},t_{33}^{[2]})=0$ and similarly for $f_0^{[1]}$ when $n=1$ (given
$f_0^{[0]}$).

Using induction, we now assume that 
 the metric $\hat h_{ij}$ is determined up to order $n+1$, i.e., that
the polynomials $f_0^{[m]}, m=0,1, \ldots n-1$ and
$f_k^{[m]}, k=1,2,\cdots 6, m=0,1, \ldots n-3$ are known ($n \geq 2$).
With these polynomials fixed, $\overline t=\overline t(\overline p)$ will be an affine function with respect to $\overline p$, and we must show that there exists a $\overline p$ such that
$\overline t(\overline p)=0$.

With $\overline t_0=\overline t(0)$, we consider the equation
\begin{equation} \label{ttuple}
\overline t(\overline p)-\overline t_0=-\overline t_0
\end{equation}
where now the mapping $\overline q: \overline p \mapsto \overline q(\overline p)=
\overline t(\overline p)-\overline t_0$ is linear. 
Explicitly, $\overline q(\overline p)=(\underline S^{[n+2]}, u_1^i u_1^j \underline T_{ij}^{[n]},u_2^i u_2^j \underline T_{ij}^{[n]},u_3^i u_3^j \underline T_{ij}^{[n]})$.
Since $\overline q$ is a mapping
\begin{equation*} \nonumber 
\overline q: \mathbf{R}^{2 n^2 +3n-2} \to \mathbf{R}^{2 n^2+14 n+24},
\end{equation*}
we have a system of equations with $2 n^2+14 n+24$ equations
and $2 n^2 +3n-2$ unknowns, i.e., the system is over-determined.

To prove that there is a $\overline p$ for which $\overline q(\overline p)=-\overline t_0$, we will prove
i) that the dimension of the range of $\overline t$, and therefore $\overline q$, is only $2 n^2 +3n-2$, and ii)
the mapping $\overline p \to \overline q(\overline p)$ is injective.

To do this, we start with the somewhat surprising\footnote{For instance, 
$u_1^i u_1^j \hat R_{ij}$ is not congruent to $0 \pmod{r^2}$.}
 lemma
\begin{lemma} \label{reddim1}
$S \equiv 0, (\textrm{ mod } r^2), \quad t_{ii} \equiv 0,  \pmod {r^2}, \ i=1,2,3$
\end{lemma}
\begin{proof}
From (\ref{S2}), using that $D(\ln|\hat h|)\equiv 0, \pmod {r^2}$, it immediately follows
that $S \equiv 0, \pmod {r^2}$. The other properties are more tedious to prove, and we refer
to appendix A for these calculations.
\end{proof}
Since all of $S, t_{11}, t_{22}, t_{33}$ contains the factor $r^2$, we can write
$$
S=r^2 \sigma, \quad t_{ii}=r^2 \tau_{ii}, \ i=1,2,3.
$$
As a consequence, $S^{[n+2]}$ is replaced by $\sigma^{[n]}$, and similarly, 
$t_{ii}^{[n+2]}$ is replaced by $\tau_{ii}^{[n]}$, and we 
consider the quadruple
\begin{equation*} \nonumber 
\overline{\tilde t}= (\sigma^{[n]}, \tau_{11}^{[n]},\tau_{22}^{]n]},\tau_{33}^{[n]})
\in \mathbf{R}^{2 n^2+6 n+4}
\end{equation*}
where we now have to solve the equation $\overline{\tilde t}(\overline p)=0$.
As before, we can put $\overline{\tilde t}_0=\overline{\tilde t}(0)$ (=$ \overline{t}(0)/r^2$)
and address the equation $\overline{\tilde q}(\overline p)=-\overline{\tilde t}_0$
where $\overline{\tilde q}$ is the linear function 
$\overline{\tilde q}(\overline p)=\overline{\tilde t}(\overline p)-\overline{\tilde t}_0$, which is a 
mapping $ \mathbf{R}^{2 n^2 +3n-2} \to \mathbf{R}^{2 n^2+6 n+4}$.
To find a $\overline p$ so that $\overline{\tilde q}(\overline p)=0$ therefore gives $2 n^2+6 n+4$
equations in  $2 n^2 +3n-2$ unknowns, i.e., the system of equations
 is still over-determined. However, there are additionally $3n+6$ linear relations
among $(\tau_{11}^{[n]},\tau_{22}^{[n]},\tau_{33}^{[n]})$ as the following arguments show.
First, Lemma \ref{reddim1} showed that
\begin{equation*} \nonumber
\begin{array}{c}
t_{11}=y^2 T_{11}-2 x y T_{12}+x^2 T_{22}=(x^2+y^2+z^2)\tau_{11}(x,y,z)\\
t_{22}=z^2 T_{22}-2 y z T_{23}+y^2 T_{33}=(x^2+y^2+z^2)\tau_{22}(x,y,z)\\
t_{33}=x^2 T_{33}-2 x z T_{13}+z^2 T_{11}=(x^2+y^2+z^2)\tau_{33}(x,y,z)
\end{array}
\end{equation*}
Putting $x=0$, the first and third relation give
$$
y^2 T_{11}(0,y,z)=(y^2+z^2)\tau_{11}(0,y,z),\quad z^2 T_{11}(0,y,z)=(y^2+z^2)\tau_{33}(0,y,z)
$$
which implies
\begin{equation} \label{reddim2}
z^2 \tau_{11}(0,y,z)=y^2 \tau_{33}(0,y,z).
\end{equation}
With the ansatz $\tau_{11}(0,y,z)=\sum_{k=0}^n a_{11,k} y^k z^{n-k}$, $\tau_{33}(0,y,z)=\sum_{k=0}^n a_{33,k} y^k z^{n-k},$
(\ref{reddim2}) implies that
$a_{11,0}=a_{11,1}=0$, $a_{11,k}=b_{33,k-2}, k=2,3,\ldots n$, $a_{33,n-1}=a_{33,n}=0$, i.e, $n+3$
linear relations. Putting $y=0$ produces another $n+3$ linear relations, as does $z=0$. However, it is easily seen
that three relations are counted twice, which menas that the total number of linear relations are $3n+6$.

Thus, the range of $\overline{\tilde q}$ is a vector space with dimension $2 n^2+3n-2$ 
(in which $-\overline{\tilde t}_0$ lives) and hence we have shown that the equation
(\ref{ttuple}) has a solution if the mapping $\overline q$ 
 is injective. This is the content of 
Lemma \ref{injektiv}.
\begin{lemma} \label{injektiv}
Let the linear mapping $\overline q: \mathbf{R}^{2 n^2 +3n-2} \to \mathbf{R}^{2 n^2 +6n+4}$
be defined by
$\overline q(\overline p)=(\underline S^{[n+2]}, u_1^i u_1^j \underline T_{ij}^{[n]},u_2^i u_2^j \underline T_{ij}^{[n]},u_3^i u_3^j \underline T_{ij}^{[n]})$. Then $\overline q$ is injective.
\end{lemma}
\begin{proof} Assume that  $\overline q(\overline p)=0$.
Consider the matrix $\gamma_{ij}$ in the decomposition (\ref{met2}) of $\hat h_{ij}$.
By splitting\footnote{This splitting is possible since $\gamma_{ij}x^j=0$.} $\gamma_{ij}=\tilde \gamma_{ij}+\frac{1}{2}\gamma r^2d\Omega_{ij}$ where $\tilde \gamma_{ij}$
is trace free, $\gamma=[\gamma_{ij}]$, and where $d\Omega_{ij}$ is the metric of the unit sphere  $S^2$, we see that
\begin{equation*} \nonumber
\hat h_{ij}^{[n+2]}=f_0^{[n]} (x,y,z)(A_1)_{ij}+r^2 \tilde 
\gamma_{ij}^{[n]}+r^2 \frac{1}{2} \gamma^{[n]} r^2d\Omega_{ij}
\end{equation*}
and in particular that
\begin{equation*} \nonumber
[\hat h_{ij}^{[n+2]}]=-2 r^2 f_0^{[n]} (x,y,z)+r^2  \gamma^{[n]} 
\end{equation*}
The first 'component' in
$\overline q(\overline p)$ 
(i.e., $S^{[n+2]}$) being zero therefore implies
$$
\gamma^{[n]}=\frac{1}{8}[12 f_0^{[n]}+\Delta_C (r^2 f_0^{[n]})-2 D( f_0^{[n]})].
$$
With the decomposition $\Delta_C=\frac{1}{r^2}\frac{\partial}{\partial r}(r^2 \frac{\partial}{\partial r})+
\frac{1}{r^2}\Delta_S$, where $\Delta_S$ is the angular Laplacian on the unit sphere, this can also be written
\begin{equation}  \label{gamman}
\gamma^{[n]}=\frac{1}{8}[(n^2+3 n+18) f_0^{[n]}+\Delta_S f_0^{[n]}].
\end{equation}
We now express $\tilde \gamma_{ij}^{[n]}$ in terms of spherical coordinates $(r,\theta,\phi)$,
and get
$$
\tilde \gamma_{ij}^{[n]}dx^i dx^j=
\left(
dr \ d\theta \ d\phi
\right)
\left(
\begin{array}{ccc}
0 & 0 & 0\\
0 & r^{n+2} f_{\tilde\gamma}(\theta,\phi) & r^{n+2} g_{\tilde\gamma}(\theta,\phi)\\
0 &	r^{n+2} g_{\tilde\gamma}(\theta,\phi) & -r^{n+2} \sin^2(\theta) f_{\tilde\gamma}(\theta,\phi)
\end{array}
\right)
\left(
\begin{array}{ccc}
dr \\ d\theta \\ d\phi\\
\end{array}
\right)
$$
for some angular functions $f_{\tilde\gamma}(\theta,\phi)$ and $g_{\tilde\gamma}(\theta,\phi)$. Note that $f_{\tilde \gamma}$ and  $g_{\tilde \gamma}$ are not arbitrary 
since $\tilde \gamma_{ij}$ must be analytic when written in the coordinates $(x,y,z)$.
With this representation, $u_k^i u_k^j \underline T_{ij}^{[n]}$ ($k=1,2,3$) are functions of $f_0^{[n]}$, $\gamma^{[n]} $, $f_{\tilde \gamma}$ and  $g_{\tilde \gamma}$.

In particular, $u_1^i u_1^j \underline T_{ij}^{[n]}=0$ gives the equation
\begin{equation} \label{u1u1T}
\begin{array}{r}
4\sin^2\theta \left[2\,  {\frac {\partial ^{2}}{\partial {\theta}^{2}}}f_{\tilde\gamma}+ (n+3) n\, f_{\tilde\gamma}  -\, {\Delta_S f_{\tilde\gamma}}\right]+
16\,\sin  \theta \cos\theta {\frac {\partial }{\partial 
\theta}}f_{\tilde\gamma} \\
+8\,{\frac {\partial ^{2}}{\partial \theta\partial \varphi }}g_{\tilde\gamma}  -2\, \sin^2  \theta
  \left(\Delta_S \gamma^{[n]}+ n ( n+1) \gamma^{[n]} \right)\\
 -  \sin^2 \theta \left( 2\,n\, \Delta_S f_0^{[n]}+{n}^{3}f_0^{[n]} - \left( n+2 \right) {\frac {\partial 
^{2}}{\partial {\theta}^{2}}}f_0^{[n]}
 \right) =0.
\end{array}
\end{equation}
The equations $u_2^i u_2^j \underline T_{ij}^{[n]}=0$ and $u_3^i u_3^j \underline T_{ij}^{[n]}=0$ are slightly
longer, namely
\begin{equation} \label{u2u2T} \hspace*{-10mm}
\begin{array}{r} 
4\,\mathcal{A}\left[ \sin^2\theta 
{\frac {\partial ^{2}f_{\tilde\gamma}}{\partial {\theta}^{2}}} +3\cos\theta \sin\theta{\frac {\partial f_{\tilde\gamma}}{\partial \theta}}-{\frac {\partial ^{2}f_{\tilde\gamma}}{\partial {\phi}^{2}}}\right]\\
  +4\, \sin^2 \theta  ( (n^2+3n)\left[\mathcal{A} -2\sin^2 \phi\right]- 4\sin^2 \phi) f_{\tilde\gamma}\\
 +8\,\mathcal{A} {\frac {\partial ^{2}g_{\tilde\gamma}}{\partial \theta\partial \phi}} -8(n^2+3n+2)\,\sin\theta\sin \phi   \cos \phi   \cos\theta g_{\tilde\gamma}\\
 -2\mathcal{A}\left[
 \sin^2 \theta {\frac {\partial ^{2}\gamma^{[n]}}{\partial {\theta}^{2}}} +\sin\theta \cos\theta  {\frac {\partial \gamma^{[n]}}{\partial \theta}}+{\frac {\partial ^{2}\gamma^{[n]}}{\partial {\phi}^{2}}}
  +n \left( n+1 \right)\sin^2 \theta \gamma^{[n]}\right]\\
 -2(n+2)\sin\theta  \sin \phi   \cos \phi   \cos \theta {\frac {\partial ^{2}f_0^{[n]}}{\partial \theta\partial \phi}}\\
  -\sin^2 \theta  \left((n-2)\mathcal{A} +(n+2)\sin^2 \phi   
 \right) {\frac {\partial ^{2}f_0^{[n]}}{ \partial {\theta}^{2}}}\\
+ \left((n+2)\sin^2 \phi-2n\mathcal{A} \right) {\frac {\partial ^{2}f_0^{[n]}}{\partial {\phi}^{2}}}
 +2(n+2)\sin \phi\cos\phi \cos^2\theta{\frac {\partial f_0^{[n]}}{\partial \phi}}\\
  -\sin\theta  \cos \theta  \left( (n-2)\sin^2\phi+2n\, \cos^2 \phi \cos^2\theta \right) {\frac {
\partial f_0^{[n]}}{\partial \theta}}
- n^3 \sin^2\theta \mathcal{A}f_0^{[n]}=0, \\
\mathcal{A}=1-\cos^2\phi \sin^2\theta
\end{array}
\end{equation}
and
\begin{equation} \label{u3u3T} \hspace*{-10mm}
\begin{array}{r}
4\mathcal{B}\left[\sin^2\theta {\frac {\partial ^{2}f_{\tilde\gamma}}{\partial {\theta}^{2}}} +3\sin\theta\cos\theta{\frac {\partial f_{\tilde\gamma} }{\partial \theta}}- {\frac {\partial ^{2}f_{\tilde\gamma}}{\partial {\phi}^{2}}} \right]\\
 +4\, \sin^2\theta  ( (n^2+3n)\left[\mathcal{B} -2\cos^2 \phi\right]- 4\cos^2 \phi)  f_{\tilde\gamma}\\
8\mathcal{B}{\frac {\partial ^{2}g_{\tilde\gamma}}{\partial \theta\partial \phi}}+8(n^2+3n+2)\,\sin\phi \cos\phi \cos\theta \sin\theta g_{\tilde\gamma}\\
-2\mathcal{B}\left[ \sin^2\theta {\frac {\partial ^{2}\gamma^{[n]}}{\partial {\theta}^{2}}} +\sin \theta \cos\theta{\frac {\partial \gamma^{[n]}}{\partial \theta}} +{\frac {\partial^{2}\gamma^{[n]}}{\partial {\phi}^{2}}} +n \left(n+1 \right) \sin^2\theta \gamma^{[n]}
\right]\\
  +2(n+2)\sin\theta \sin \phi \cos\phi \cos \theta {\frac {\partial ^{2}f_0^{[n]}}{\partial \theta\partial\phi}}\\
   - \sin^2\theta \left(  (n-2)\mathcal{B} +(n+2)\cos^2\phi \right) {\frac {\partial ^{2}f_0^{[n]}}{\partial {\theta}^{2}}}\\
 + \left((n+2)\, \cos^2\phi -2\,n \mathcal{B} \right){\frac{\partial ^{2}f_0^{[n]}}{\partial {\phi}^{2}}}
-2(n+2)\sin \phi \cos \phi \cos^2 \theta{\frac {\partial f_0^{[n]} }{\partial \phi}}\\
 -\sin\theta\cos\theta \left((n-2)\cos^2\phi+2 n\sin^2\phi\cos^2\theta  \right) {\frac {\partial f_0^{[n]}}{\partial \theta}} -n^3 \sin^2\theta  \mathcal{B}f_0^{[n]}=0,\\
 \mathcal{B}=1- \sin^2\phi  \sin^2\theta
 \end{array}
\end{equation}

$\gamma^{[n]}$ can be eliminated using (\ref{gamman}), and after insertion in (\ref{u1u1T}),
one can express ${\frac {\partial ^{2}g_{\tilde\gamma} }{\partial \theta\partial \varphi }}$
in terms of $f_0^{[n]}$ and $f_{\tilde \gamma}$.
With $\gamma^{[n]}$ and ${\frac {\partial ^{2}g_{\tilde\gamma} }{\partial \theta\partial \varphi }}$
inserted in (\ref{u2u2T}), one can 'solve' for $g_{\tilde\gamma} $, i.e., express
also $g_{\tilde\gamma} $ in terms of $f_0^{[n]}$ and $f_{\tilde \gamma}$.
From this,  (\ref{u3u3T}) gives $f_{\tilde\gamma}$, and therefore also
$g_{\tilde\gamma}$ and  $\gamma^{[n]}$ as a function of $f_0^{[n]}$.
Finally, we insert all these quantities in $u_1^i u_1^j \underline T_{ij}^{[n]}$. Using the obtained expression for ${\frac {\partial ^{2}g_{\tilde\gamma} }{\partial \theta\partial \varphi }}$ gives a trivial identity, but by using the expression for $g_{\tilde\gamma}$, differentiated with respect to $\theta$ and $\varphi$, we get a fourth order linear equation
for $f_0^{[n]}$. We now put $f_0^{[n]}=r^n f_a(\theta,\phi)$, which in particular means that
$f_a(\theta,\phi)$ is a linear combination of spherical harmonics $Y_l^m(\theta,\phi)$
with $l \leq n$. After these steps,
the equation for the angular part $f_a(\theta,\phi)$ is found to be
\begin{equation}
\begin{array}{r}
 \sin^2 \theta {\frac {\partial ^{4} f_a}{\partial {\theta}^{4}}}
 +2\,{\frac {\partial ^{4} f_a}{\partial {\theta}^{2}\partial {\varphi }^{2}}}
 +{\frac {{\frac {\partial ^{4} f_a}{\partial {\varphi }^{4}}}}{ \sin^2\theta}}
 +2\,\cos\theta \sin \theta {\frac {
\partial ^{3}f_a}{\partial {\theta}^{3} }} 
 -2\,{\frac {\cos \theta {\frac {\partial ^{3}f_a}
{\partial \theta\partial {\varphi }^{2}}}}{\sin\theta}}\\
- \left( 
(2n^2+6n+5) \cos^2\theta -2n^2-6n-4 \right) {\frac {\partial ^{2}f_a}{
\partial {\theta}^{2}}}\\
 -2\,{\frac { \left( n^2+3n+2) \cos^2\theta
-n^2-3n-4 \right) {\frac {
\partial ^{2}f_a}{\partial {\varphi }^{2}}} }{ \sin^2\theta}}\\
-{\frac {\cos\theta \left( (2n^2+6n+6)
 \cos^2\theta-2n^2-6n-7 \right) {\frac {\partial f_a}{\partial 
\theta}} }{\sin \theta }}\\
 +\sin^2\theta \left( n+
3 \right)  \left( n+2 \right) \left( n+1 \right)n\,f_a=0.
\label{slutdifftvavar}
\end{array}
\end{equation}
Since the coefficients of (\ref{slutdifftvavar}) do not contain $\phi$, we make the ansatz
$f_a(\theta,\phi)=F(\theta)e^{i m\phi}$.  The equation for $F(\theta)$ becomes
\begin{equation}
\begin{array}{r}
\sin^2\theta{\frac {d^{4}}{d{\theta}^{4}}}F \left( \theta \right) 
+2\,\cos\theta\sin\theta {\frac {d^{3}}{d{\theta}^{3}}}F \left( \theta \right)\\
-[2 m^2+
 \left( (2n^2+6n+5) \cos^2\theta-2n^2-6n-4 \right) ]{\frac {d^{2}}{d{\theta}^{2}}}F \left( \theta\right)\\
  +[2m^2
   -\left((2n^2+6n+6) \cos^2\theta -2n^2-6n-7 \right)]\frac{\cos\theta}{\sin\theta}{\frac {d}{d\theta}
}F \left( \theta \right) \\
+[ {\frac {m^4+2m^2 \left( (n^2+3n+2) \cos^2\theta-n^2-3n-4 \right)}{ \sin^2\theta}} \\
+ \sin^2\theta (n+3) (n+2)(n+1)n ]F \left( \theta \right) =0
\label{slutdiffenvar}
\end{array}
\end{equation}
The solution to (\ref{slutdiffenvar}) is
\begin{equation}
\begin{array}{r}
 F \left( \theta \right) =
 {\it C_1}\,\sin\theta {\it P_{n+1}^{m-1}} \left(\cos\theta\right) 
+{\it C_2}\,\sin\theta {\it P_{n+1}^{m+1}} \left(\cos\theta\right) \\
 +{\it C_3}\,\sin  \theta{\it Q_{n+1}^{m-1}} \left(\cos\theta \right) 
 +{\it C_4}\,\sin \theta{\it Q_{n+1}^{m+1}}\left(\cos\theta\right) 
\label{slutdiffensol}
\end{array}
\end{equation}
i.e., a linear combination of associated Legendre functions of both kinds. Now, since $f_a(\theta,\phi)$
consists of spherical harmonics $Y_l^m(\theta,\phi)$
with $l \leq n$, (\ref{slutdiffensol}) implies that the relation $f_a(\theta,\phi)=F(\theta)e^{i m\phi}$
requires that $F$ and hence $f_a$ is identically zero.
From this,  (\ref{gamman}) gives $\gamma^{[n]}=0$. By forming 
$(1+\cos^2\theta) u_1 u_1 \underline T_{ij}^{[n]}-(u_2 u_2 \underline T_{ij}^{[n]}+u_3 u_3 \underline T_{ij}^{[n]})
\sin^2\theta$
it is seen that $f_{\tilde \gamma}=0$,
from which $g_{\tilde \gamma}$ must also be zero.
This means that the mapping $\overline q$ is injective, and therefore that there exists
a unique formal series expansion of $\hat h_{ij}$ such that $S=0, t_{11}=0,
t_{22}=0, t_{33}=0$ (Lemma \ref{f1})
\end{proof}
We again stress that  (\ref{ttuple})  gives the metric components of the
appropriate order explicitly.

\subsection{Convergence of the metric}\label{conv}
In this section we will prove that the series expansion for $\hat h_{ij}$, which was found in the previous section, converges and also solves the full conformal field equations. 
This is the content of the following lemma.
\begin{lemma} \label{f2}
Suppose that $\hat h_{ij}$ is a formal power series for the metric of the form {\rm(\ref{metricform})}, producing the moments of Theorem \ref{thm}.   If
$$
S=0, \quad t_{11}=0, \quad t_{22}=0,\quad  t_{33}=0,
$$
then 
$$
t_{00}=0, \quad t_{0k}=0, k=1,2,3
$$
and the power series is convergent in a neighbourhood of $i^0$.
\end{lemma}
This will conclude the proof of Theorem \ref{thm} when the monopole is nonzero.
\begin{proof}
To prove this lemma, we will use a result by Friedrich, namely Theorem 1.1 in \cite{friedrich}, where
 a related property is investigated.
In \cite{friedrich} it is proven that under rather similar settings, there exists a metric
which instead of prescribed multipole moments, produces a set of prescribed {\em null data}, 
where also growth conditions for the null data in order for the series expansion of the metric
to converge are given.
These null data are proven to be in a one-to-one correspondence with the family of multipole
moments with non-zero monopole, but since this correspondence is rather implicit, the actual conditions on the multipole
moments (as required by the conjecture by Geroch) are not clear.

We vill prove that the conditions on the moments $P^0, P^0_{i_1}, P^0_{i_1 i_2}, \cdots$
in Theorem \ref{thm} can be carried over to estimates on the null data as required in
\cite{friedrich}. This will then guarantee a solution to the conformal Einstein's field equations
which according to the work presented here will have the desired multipole moments. 

In order to be able to refer to the work in \cite{friedrich}, we continue to impose the temporary condition
that the monopole $P^0$ is non-zero.

To begin, we compare the different conformal settings used. In \cite{friedrich},
one uses the conformal factor 
$$
\Omega_F=\left(\frac{1-\sqrt\lambda}{m}\right)^2
$$
while this work uses (starting with $(h_G)_{ab}$ in Section \ref{tfe}) the conformal
factor
$$
\hat \Omega=\sqrt \lambda \Omega=\sqrt \lambda r^2 e^{2\kappa}
$$
Therefore, if $(h_F)_{ab}$ denotes the metric used in \cite{friedrich}, we have the
relation
$$
(h_F)_{ab}=\left[ 
\frac{1}{m^2 \Omega } \frac{(1-\sqrt{\lambda})^2}{\sqrt{\lambda}}
\right]^2 \hat h_{ab}, \quad
\lambda=\sqrt{1+4 \Omega \hat \phi^2}+2 \sqrt{\Omega}\hat \phi,
\quad \Omega = r^2 e^{2\kappa}.
$$
Since $\sqrt \Omega$ contains the non-smooth quantity $r$, it is not obvious that the transition
from $\hat h_{ab}$ to $(h_F)_{ab}$ is smooth. However, by putting $\xi= 2 r e^\kappa \hat \phi $,
the conformal factor $\Omega_{T}$ relating $\hat h_{ab}$ and $(h_F)_{ab}$ is seen to be
$$
\Omega_T=\frac{1}{m^2 \Omega } \frac{(1-\sqrt{\lambda})^2}{\sqrt{\lambda}}
=
\frac{4 \hat \phi^2}{m^2} 
 \frac{\left(1-\sqrt{\sqrt{1+\xi^2}+\xi}\right)^2}{\xi^2 \sqrt{\sqrt{1+\xi^2}+\xi}}.
$$
By using $\sqrt{1+\xi^2}-\xi=\frac{1}{\sqrt{1+\xi^2}+\xi}$, it follows that this expression is even in
$\xi$, which means that only even powers of $r$ occurs in $\Omega_T$. Since the limit 
$ \lim_{r \to 0} \Omega_T=1$ causes no problem, $\Omega_T$ is found to be (formally) smooth.
Note that since we know that $\hat \phi$ produces the prescribed moments under
$\hat h_{ab}$, the potential $\hat \phi/\sqrt{\Omega_T}$ will give the correct moments
when using $(h_F)_{ab}$.

Now, suppose that 
$
S=0,  t_{11}=0,  t_{22}=0,  t_{33}=0,
$
and hence, according to Lemma \ref{f1},  that we have a formal metric $\hat h_{ij}$. Going over to 
$(h_F)_{ab}$, and using the trace-free Ricci tensor, this defines a set of abstract null data
($\mathcal D_n$ and $\mathcal D_n^*$ below).
By the arguments in \cite{friedrich},  there exists a formal solution to the full conformal static field equations
connected to these null data, and by going back to $\hat h_{ij}$, this implies that also
$
t_{00}=0, t_{0k}=0, k=1,2,3
$.

Another way of putting this is: if $(h_F)_{ab}$ satisfies the conformal static field equations and
if $(h_F)_{ab}$ and $\hat h_{ab}$ are conformally related, the equations (\ref{confeq})
are automatically satisfied,
 while the equations $S=0, t_{11}=0, t_{22}=0, t_{33}=0$ fully
determines the metric components $\hat h_{ij}$ in terms of the normal coordinates ${x,y,z}$
(taking the special form if $\hat h_{ij}$ into account).

In \cite{friedrich}, one uses  normal coordinates around $i^0$,
and  introduces a frame field $c_{\mathbf a}, \mathbf a=1, 2, 3$ which is parallely 
propagated along the geodesics through $i^0$. Since we write $\hat R_{ij}$ for the Ricci tensor,
we can let $R_{ab}$ stand for the Ricci tensor in \cite{friedrich}.
With $s_{ab}$ denoting the trace free part of the Ricci tensor, we have $s_{ab} =R_{ab}$,
since $R_{ab}$  in \cite{friedrich} already is trace free due to the choice of conformal gauge there.
One then consider the set
$$
\mathcal D_n=\{
s_{a_2 a_1}(i^0), 
C[\nabla_{a_3}s_{a_2 a_1} ](i^0),
C[\nabla_{a_4}\nabla_{a_3}s_{a_2 a_1} ](i^0), \ldots
\}
$$
where $\nabla_a$ is the derivative operator associated with $(h_F)_{ab}$.
To proceed, one then expresses the family of tensors in $\mathcal D_n$
in the introduced frame field, and considers (again, see \cite{friedrich}  for the details)
the related family
$$
\mathcal D_n^*=\{
s_{\mathbf{a_2 a_1}}(i^0), 
C[\nabla_{\mathbf{a_3}}s_{\mathbf{a_2 a_1}} ](i^0),
C[\nabla_{\mathbf{a_4}}\nabla_{\mathbf{a_3}}s_{\mathbf{a_2 a_1}} ](i^0), \ldots
\}
$$
According to Theorem 1.1 in \cite{friedrich}, there exists an analytic solution around $i^0$
to the conformal static vacuum field equations with $m \neq 0 $ for the conformal metric $(h_F)_{ab}$ 
to each set of tensors, given in the orthonormal frame (at $i^0$),
$$
\hat {\mathcal D}_n=\{
\psi_{\mathbf{a_2 a_1}}, 
\psi_{\mathbf{a_3a_2 a_1}},
\psi_{\mathbf{a_4 a_3 a_2 a_1}}, \ldots
\}
$$
which satisfy: i) each tensor is totally symmetric and trace free, and ii) there exist constants $M,r>0$
such that 
\begin{equation}
|\psi_{\mathbf{a_p \ldots a_1 b c}}| \leq \frac{M p!}{r^p}, \quad
a_p, \ldots a_1, b, c =  1,2,3, \quad p=0,1,2, \ldots
\label{tensest}
\end{equation}
In particular, in the introduced frame, one has
\begin{equation}
C[\nabla_{\mathbf{a_q}} \ldots \nabla_{\mathbf{a_3}}s_{\mathbf{a_2 a_1}} ](i^0)
=\psi_{\mathbf{a_q \ldots a_1}}, \quad q=2,3,4, \ldots
\label{slikhet}
\end{equation}
We will show that given the conditions of Theorem \ref{thm}, there will exist such a 
family $\hat {\mathcal D}_n$ with the desired properties and, most importantly, 
 that
the tensors in $\hat {\mathcal D}_n$ produce the prescribed moments via the requirement
(\ref{slikhet}).
We have already concluded that given $\hat h_{ij}$, $\Omega_T$ and $\hat \phi$, this defines
 $(h_F)_{ab}$, and moreover that $\hat \phi/\sqrt{\Omega_T}$ will give the correct moments
when using $(h_F)_{ab}$. We can therefore define 
$\psi_{\mathbf{a_q \ldots a_1}}$ in (\ref{slikhet}) as the corresponding left hand side and prove the necessary
estimates on $C[\nabla_{\mathbf{a_q}} \ldots \nabla_{\mathbf{a_3}}s_{\mathbf{a_2 a_1}} ](i^0)$
directly.

Now, with the notation in Theorem \ref{thm}, the convergence of 
$\sum_{|\alpha|\geq 0}\frac{\mathbf{r}^\alpha}{\alpha!}P^0_\alpha$ near $i^0$ 
implies that for some constants $M,r>0$,
\begin{equation}
|P^0_\alpha| \leq \frac{p! M}{r^p}, \quad p=0,1,2,\ldots
\label{ptensorestimates}
\end{equation}
From the arrangement in \cite{friedrich}, $P_{a_1 a_2}=-\frac{m}{2}s_{a_1 a_2}$, 
and hence the estimates in (\ref{tensest}) will be satisfied if there are
constants $\hat M, \hat r$ so that
\begin{equation*}
|C[\nabla_{\mathbf{a_q}} \ldots \nabla_{\mathbf{a_3}}P_{\mathbf{a_2 a_1}} ](i^0)|
 \leq \frac{|\alpha| ! \hat M}{\hat r^{|\alpha|}}, \quad |\alpha|=p=0,1,2,\ldots
\nonumber 
\end{equation*}

Proceeding as in \cite{friedrich}, the tensors $P_\alpha$ can be expressed
as totally symmetric space spinors fulfilling a certain reality condition [ (3.4) in
\cite{friedrich} ]. In terms of these space spinors, (\ref{ptensorestimates}) reads
\begin{equation}
|P^0_{A_p B_p \ldots A_1 B_1}| \leq \frac{p ! \tilde M}{\tilde r^p},
\quad A_p, B_p, \ldots A_1, B_1 =0,1, 
 \quad  p=0,1,2, \ldots
\label{pspinorestimates}
\end{equation}
for some related constants $\tilde M, \tilde r$,
and we must demonstrate that given the estimates (\ref{pspinorestimates}),
there are $\check M, \check r$ so that
\begin{equation*} \hspace*{-15mm}
| D_{(A_p B_p } \ldots D_{A_3 B_3 }P_{A_2 B_2 A_1 B_1)}(i^0)| \leq \frac{p ! \check M}{\check r^p},
\ A_p, B_p, \ldots A_1, B_1 =0,1, 
 \  p=0,1, \ldots
\nonumber 
\end{equation*}
This will be done using induction, and to get the arguments to work, we first observe that
given (\ref{pspinorestimates}), it is easy to find constants $M_0, r_0$ so that
\begin{equation}\hspace*{-15mm}
|P_{A_p B_p \ldots A_1 B_1}(i^0)| \leq \frac{(p-1) !  M_0}{r_0^p},
\quad A_p, B_p, \ldots A_1, B_1 =0,1, 
 \quad  p=1,2, 3\ldots
\label{pspinorestimates2}
\end{equation}
To simplify the calculations, we introduce the following notation. 
We denote $P_{A_p B_p \ldots A_1 B_1}$  by $P_p$, and let
$D_{(A_p B_p } \ldots D_{A_{k+1} B_{k+1} }P_{A_k B_k \ldots A_2 B_2 A_1 B_1)}$
be denoted by $D_{p-k}P_k$. In terms of this notation, the recursion (\ref{orgrec})
expressed in space spinors becomes ($c_p=\frac{p(2p-1)}{2}$)
\begin{equation}
P_p=D_1 P_{p-1}- c_{p-1}R_2 P_{p-2}
\label{frec1}
\end{equation}
where $R_2$ stands for $R_{A_1 B_1 A_2 B_2}$ and it is understood that
all products involves complete symmetrization. For instance, $R_2 P_{p-2}$ stands
for \newline
$R_{(A_p B_p A_{p-1} B_{p-1} }P_{A_{p-2} B_{p-2} \ldots A_2 B_2 A_1 B_1)}$.
Since $P_{a b}=-\frac{m}{2}R_{ab}$, a simple rescaling
$P_{a_p \ldots a_1} \to -\frac{m}{2} P_{a_p \ldots a_1}$ (keeping the same stem letter),
transforms (\ref{frec1}) into
\begin{equation*}
D_1 P_{p}=P_{p+1}+c_{p}P_{p-1}P_2 
\nonumber 
\end{equation*}
We now claim that 
\begin{equation}
\forall n \geq 0, |D_n P_{p}(i^0)| \leq \frac{M_0(1+2 M_0)^n (n+p-1)!}{r_0^{n+p}}, \quad p=2,3,4, \ldots
\label{recest}
\end{equation}
where (\ref{pspinorestimates2}) is the initial estimate for $n=0$, and where we assume that
(\ref{recest}) is valid up to a certain value of $n$. For $n+1$ we then get
\begin{equation*} \nonumber
D_{n+1}P_p=D_n P_{p+1}+c_p \sum_{k=0}^n 
\binom{n}{k}
D_k P_{p-1}D_{n-k} P_2,
\end{equation*}
and
inserting the estimimates (\ref{recest}) we get, where (\ref{estb}) to (\ref{ests}) are evaluated at
$(i^0)$,
\begin{equation}\hspace*{-5mm}
\begin{array}{cl}
|D_{n+1}P_p| &
 \leq |D_n P_{p+1}|+c_p \sum_{k=0}^n 
\binom{n}{k}
|D_k P_{p-1}| |D_{n-k} P_2| \\
&
 \leq \frac{M_0(1+2 M_0)^n (n+p)!}{r_0^{n+p+1}} 
 +c_p \sum_{k=0}^n 
\binom{n}{k}
\frac{M_0^2 (1+2 M_0)^n (k+p-2)!(n-k+1)!}{r_0^{n+p+1}}.
\label{estb}
\end{array}
\end{equation}
Next, from $c_p \leq 2 p(p-1)$, for $p \geq 2$, we get
\begin{equation}\hspace*{-25mm}
|D_{n+1}P_p| 
  \leq  \frac{M_0(1+2 M_0)^n}{r_0^{n+p+1}}\left[  (n+p)! 
+2 M_0 p(p-1) \sum_{k=0}^n 
\binom{n}{k}
 (k+p-2)!(n-k+1)! \right]
\end{equation}
However, from the identity
$$
p(p-1) \sum_{k=0}^n 
\binom{n}{k}
 (k+p-2)!(n-k+1)! =(n+p)!,
$$
we get
\begin{equation}\hspace*{-20mm}
|D_{n+1}P_p| 
  \leq  \frac{M_0(1+2 M_0)^n}{r_0^{n+p+1}}\left[  (n+p)! 
+2 M_0 (n+p)! \right]=
\frac{M_0(1+2 M_0)^{n+1}(n+p)!}{r_0^{n+p+1}}
\end{equation}
as claimed in (\ref{recest}). Thus the estimates in (\ref{recest}) are valid, and by putting
$p=2$, we find that (at $i^0$)
\begin{equation}
\forall n \geq 0, |D_n P_{2}| \leq \frac{M_0(1+2 M_0)^n (n+1)!}{r_0^{n+2}}
\label{ests}
\end{equation}
Since it is easy to find constants $\check M, \check r$ such that
$\frac{M_0(1+2 M_0)^n (n+1)!}{r_0^{n+2}} \leq \frac{n! \check M}{\check r^n}$
for $n=0, 1, 2, \cdots$, Theorem 1.1 of \cite{friedrich} follows.

Finally, to see that also the series expansion for $\hat h_{ij}$ converges we note that, 
near $i^0$, Theorem 1.1 of \cite{friedrich} shows that the physical 3-metric $g_{ij}$ exists.
This means that $\lambda (h_G)_{ij} =\Omega^{-2}\hat h_{ij}=\frac{e^{-4\kappa}}{r^4}\hat h_{ij}$ exist as functions and in particular (after transvecting with $x^j$) that $\frac{e^{-4\kappa}}{r^4}x^i$ exist. This implies that the series expansions for $\kappa$, and thus for $\Omega$, and therefore also for $\hat h_{ij}$, converges near $i^0$.
\end{proof}

\subsection{The case $P^0=0$} \label{massless}
To complete the proof, we must now relax the condition $m \neq 0$, and this will be possible due to the fact that the conformal factor $\hat \Omega$ is well behaved for $m=0$. Thus, suppose that
we have a sequence of totally symmetric and trace-free tensors
$P^0, P^0_{i_1}, P^0_{i_1 i_2}, \ldots$ as in Theorem \ref{thm}, and where $P^0=0$.
In particular, it is assumed that $u(\mathbf{r})=\sum_{|\alpha|\geq 0}\frac{\mathbf{r}^\alpha}{\alpha!}
P^0_\alpha $ converges in some polydisc 
$U: \{(x,y,z), |x|<d, |y|<d, |z[<d\}$. 
By replacing only the monopole, i.e.,
putting $P^0=m_0 > 0$, the corresponding sequence $m_0, P^0_{i_1}, P^0_{i_1 i_2}, \ldots$
corresponds to the function $\tilde u(\mathbf{r})=m_0+\sum_{|\alpha|> 0}\frac{\mathbf{r}^\alpha}{\alpha!}
P^0_\alpha $, which also converges in $U$. Thus, by the arguments given so far
there exists convergent, in a polydisc $V$ say,  power series 
$$\hat h_{ij}=\sum_{|\alpha| \geq 0} (c_{ij})_\alpha \mathbf{r}^\alpha$$
for the metric components; furthermore we also know that there is a static spacetime having the multipole moments $m_0, P^0_{i_1}, P^0_{i_1 i_2}, \ldots$. Now, from the recursion producing the
metric, it is seen that each coefficient $(c_{ij})_\alpha$ is a polynomial in $m_0$, and hence
the metric components $h_{ij}$ can be regarded as a power series in the four variables
$(m,x,y,z)$:
$$\hat h_{ij}=\sum_{|\beta| \geq 0} (d_{ij})_\beta  m^{\beta_1} x^{\beta_2} y^{\beta_3} z^{\beta_4},$$
where the multi-index $\beta=(\beta_1,\beta_2,\beta_3,\beta_4)$. Since this series converges for
$m=m_0$, $(x,y,z) \in V$, it also converges for $|m|<m_0, (x,y,z)\in V$. We can now choose
$m=0$ and still have convergence of $\hat h_{ij}$, $\kappa$ and $g_{ij}$. In particular, with $m=0$
the multipole moments will be the initially desired sequence
$P^0, P^0_{i_1}, P^0_{i_1 i_2}, \ldots$ where $P^0=0$.

\section{Acknowledgement}
The author wants to thank Prof. Brian Edgar for helpful comments.

\section{Discussion}
In this paper, we have proved that the necessary conditions from \cite{bahe} on the multipole moments
of a static spacetime are also sufficient in order for a static vacuum spacetime having these moment
to exist. In particular, we do not assume that the monopole is nonzero. This proves a long standing conjecture by Geroch, \cite{geroch}.

The conditions given are simple and natural. In essence, each allowed set of multipole moments is
connected to a harmonic function on $\mathbf R^3$, defined in a neighbourhood of $\mathbf 0$.
The proof is constructive in the sense that an explicit metric having prescribed
moments up to a given order can be calculated.

Considering future work, it is natural to see what can be carried over to the general
stationary case. It may be conjectured that the natural generalisation of Theorem \ref{thm}
is valid, with the extra condition that the monopole is real and nonzero. 
It could also be instructive to explicitly calculate the metric belonging
to a static spacetime with arbitrary multipole moments up to a given order.

\section{Appendix}
\begin{proof}[Proof of Lemma \ref{gdec}]
Again we make an ansatz
\begin{equation*} \nonumber 
A=\left(
\begin{array}{ccc}
X(x,y,z) & a(x,y,z) & b(x,y,z)\\
a(x,y,z) & Y(x,y,z) & c(x,y,z)\\
b(x,y,z) & c(x,y,z) & Z(x,y,z)\end{array}
 \right)
\end{equation*}
where $a,b,c,X,Y$ and $Y$ are analytic in the variables indicated. The condition $A_{ij}x^i=0$
is
\begin{equation}\label{Sref}
\begin{array}{ccc}
x X(x,y,z) +y  a(x,y,z) +z b(x,y,z) &=&0\\
x a(x,y,z) +y Y(x,y,z) +z c(x,y,z)&=&0\\
x b(x,y,z) +y  c(x,y,z) +z Z(x,y,z)&=& 0
\end{array}
\end{equation}
Writing
$$
X(x,y,z)=X_2(x,y)+z\ X_3(x,y,z), \quad a(x,y,z)=a_2(x,y)+z\ a_3(x,y,z),
$$
(\ref{Sref}) with $z=0$ shows that $a_2(x,y)=x\ a_4(x,y)$, $X_2(x,y)=-y\ a_4(x,y)$ for some function $a_4$. Next it follows that $ b(x,y,z)=-y\ a_3(x,y,z)-x\ X_3(x,y,z)$. After the decomposition
$Y(x,y,z)=Y_2(x,y)+z\ Y_3(x,y,z)$, (\ref{Sref}) with $z=0$ shows that
$a_4(x,y)=y\ f_1(x,y)$, $Y_2(x,y)=-x^2 f_1(x,y)$ for some function $f_1$. Form this we find that
$c(x,y,z)=-x\ a_3(x,y,z)-y Y_3(x,y,z)$. Next we make the decomposition
\begin{equation*} \nonumber
\begin{array}{c}
a_3(x,y,z)=a_6(x,y)+z\ f_5(x,y,z)\\
X_3(x,y,z)=X_6(x,y)+z\ f_4(x,y,z)\\
Y_3(x,y,z)=Y_6(x,y)+z\ f_6(x,y,z)
\end{array}
\end{equation*}
form which (\ref{Sref}), again with $z=0$, shows that
$X_6(x,y)=2 y\ f_2(x,y)$, $Y_6(x,y)=2 x\ f_3(x,y)$ for some function $f_2$ and $f_3$.
Finally, it follows that $a_6(x,y)=-x\ f_2(x,y)-y\ f_3(x,y)$ and
$Z(x,y,z)=2 x y\ f_5(x,y,z)+x^2 f_4(x,y,z)+y^2 f_6(x,y,z)$.
By collecting terms, the lemma follows.
\end{proof}
\begin{proof}[Proof of Lemma \ref{reddim1}]
Observe that the definition of $u_i^a$ implies that
$u_i^a(r^2)=0,\ i=1,2,3$. We now demonstrate that
$t_{11} =u_1^a u_1^b T_{ab} \equiv 0 \pmod {r^2}$.
From the definition of $\hat R_{ij}$,
\begin{equation*} \nonumber
u_1^i u_1^j R_{i j} =\  
 u_1^i u_1^j  \frac{\partial}{\partial x^k}
{ \Gamma^k}_{i j}-u_1^i u_1^j  \frac{\partial}{\partial x^i}
{ \Gamma^k}_{k j}\\
+
 u_1^i u_1^j { \Gamma^m}_{i j}{ \Gamma^k}_{m k}
-u_1^i u_1^j { \Gamma^m}_{k j}{ \Gamma^k}_{m i}
\end{equation*}
where 
${ \Gamma^m}_{i j}= \frac{1}{2} \hat h^{m \sigma}
\{\frac{\partial \hat h_{i \sigma}}{\partial x^j} +\frac{\partial \hat h_{j \sigma}}{\partial x^i}-
\frac{\partial \hat h_{j i}}{\partial x^\sigma}     \}
$ and
${\Gamma^m}_{i m}=\frac{1}{2} \partial_i \ln|\hat h|$. 
Next, if $B$ denotes any symmetric $3 \times 3$ matrix, we can write
\begin{equation}
\hat h_{ij}= \eta_{ij}+p_0 x_i x_j +r^2 B_{ij}
\label{hdecigen}
\end{equation}
from which it follows that
$$
\hat h^{ij}= \eta^{ij}-p_0 x^i x^j +r^2 B^{ij}.
$$
Therefore, $\hat h^{m \sigma}{ \Gamma^k}_{m k} \equiv (\eta^{m\sigma}-p_0 x^m x^\sigma)\frac{1}{2} \partial_m \ln|\hat h|
\equiv \eta^{m\sigma}\frac{1}{2} \partial_m \ln|\hat h|  \pmod {r^2}$, where we have used that 
$D(\ln |h|) \equiv 0  \pmod {r^2}$.
Thus, $u_1^i u_1^j { \Gamma^m}_{i j}{ \Gamma^k}_{m k} \equiv
\frac{1}{4} u_1^i u_1^j  
\{\frac{\partial \hat h_{i \sigma}}{\partial x^j} +\frac{\partial \hat h_{j \sigma}}{\partial x^i}-
\frac{\partial \hat h_{j i}}{\partial x^\sigma} \} \eta^{m\sigma}\partial_m \ln|\hat h|  \pmod {r^2}$.
Using that $\hat h_{ij}$ can be written as in (\ref{hdecigen}) and  that $u_1^i \partial_i x^j=u_1^j$ it follows that
$u_1^i u_1^j { \Gamma^m}_{i j}{ \Gamma^k}_{m k} \equiv 0  \pmod {r^2}$. A similar but longer
calculation shows that also $u_1^i u_1^j { \Gamma^m}_{k j}{ \Gamma^k}_{m i} \equiv 0 \pmod {r^2}$.

By expanding (\ref{rab2a}), using the above result, (\ref{chidef}) and the relations (\ref{prop}), it amounts to proving
that (using the form (\ref{met2}) for the metric)
\begin{equation*} \nonumber
\begin{array}{l}
 u_1^i u_1^j 
  \left[\partial_k
{ \Gamma^k}_{i j}-\partial_i
{ \Gamma^k}_{k j}
+D(\gamma_{ij})+2\gamma_{ij}+2\hat D_i \hat D_j \kappa \right. \\
 + 
 \eta_{ij}(2\Delta \kappa-D(f)-2 f - 8D(\chi) - 16\chi + D(\ln|\hat h|)/r^2) \left. \right]
\equiv 0  \pmod {r^2}
\end{array}
\end{equation*}
Next, by replacing $\kappa=r^2\chi$,  it follows that $u_1^i u_1^j(2\hat D_i \hat D_j \kappa+\eta_{ij}(2\Delta \kappa - 8D(\chi) - 16\chi))
\equiv 0  \pmod {r^2}$ which leaves us with
\begin{equation} \hspace*{-10mm}
 u_1^i u_1^j 
  \left[\partial_k
{ \Gamma^k}_{i j}-\partial_i
{ \Gamma^k}_{k j}
+D(\gamma_{ij})+2\gamma_{ij}+
 \eta_{ij}(-D(f)-2 f  + D(\ln|\hat h|)/r^2) \right]
\label{rsimp0}
\end{equation}
With $\ln |\hat h|=r^2 \sigma$, 
we get $u_1^i u_1^j \partial_i { \Gamma^k}_{k j}=u_1^i u_1^j \frac{1}{2}\partial_i \partial_j (r^2 \sigma)
\equiv \sigma u_1^i u_1^j \eta_{ij}  \pmod {r^2}$.
Next, 
$$
2 u_1^i u_1^j\partial_k { \Gamma^k}_{i j}=
u_1^i u_1^j\partial_k [\hat h^{k \sigma}
\{\frac{\partial \hat h_{i \sigma}}{\partial x^j} +\frac{\partial \hat h_{j \sigma}}{\partial x^i}-
\frac{\partial \hat h_{j i}}{\partial x^\sigma} \}]
$$
and by using Leibniz's rule, one finds that 
\begin{equation} \hspace*{-10mm}
u_1^i u_1^j\partial_k (\hat h^{k \sigma})
\{\frac{\partial \hat h_{i \sigma}}{\partial x^j} +\frac{\partial \hat h_{j \sigma}}{\partial x^i}-
\frac{\partial \hat h_{j i}}{\partial x^\sigma} \} \equiv
-u_1^i u_1^j \partial_k (\hat h^{k \sigma}) \partial_\sigma (\hat h_{ij})\  \pmod {r^2}
\label{r2simp}
\end{equation}
and
\begin{equation}\hspace*{-10mm}
u_1^i u_1^j \hat h^{k \sigma}\partial_k 
\{\frac{\partial \hat h_{i \sigma}}{\partial x^j} +\frac{\partial \hat h_{j \sigma}}{\partial x^i}-
\frac{\partial \hat h_{j i}}{\partial x^\sigma} \} \equiv
u_1^i u_1^j [2\eta^{k\sigma} \partial_\sigma \partial_{(j} \hat h_{i)k}-\Delta_C \hat h_{ij}]  \pmod {r^2}
\label{r2simp2}
\end{equation}
Some further manipulations reveals that in (\ref{r2simp}),
$u_1^i u_1^j \partial_k (\hat h^{k \sigma}) \partial_\sigma (\hat h_{ij})\equiv 0 \pmod {r^2}$
while from (\ref{r2simp2}) we find
$u_1^i u_1^j [2\eta^{k\sigma} \partial_\sigma \partial_{(j} \hat h_{i)k}-\Delta_C \hat h_{ij}]
\equiv u_1^i u_1^j [8 f \eta_{ij}+6 D(f)\eta_{ij}-6\gamma_{ij}-4 D(\gamma_{ij})]  \pmod {r^2}$
Inserted in (\ref{rsimp0}), we get, modulo $r^2$,
\begin{equation}\hspace*{-10mm}
 u_1^i u_1^j 
  \left[2 f \eta_{ij}+2 D(f)\eta_{ij}-\gamma_{ij}-D(\gamma_{ij})-
\frac{1}{2}\partial_i \partial_j (\ln |\hat h|)+
 \eta_{ij}(D(\ln|\hat h|)/r^2) \right]
\label{nastsista}
\end{equation}
where also $u_1^i u_1^j \partial_i \partial_j (\ln |\hat h|) 
\equiv -z\partial_z (\ln |\hat h|)  \pmod {r^2}$. To finish the proof,
some linear algebra shows that $\ln|\hat h| =-2 f r^2+[\gamma_{ij}] r^2  \pmod {r^4}$.
Finally an insertion in (\ref{nastsista}) gives, modulo $r^2$,
\begin{equation}
- u_1^i u_1^j(\gamma_{ij}+D(\gamma_{ij}))-z^2([\gamma_{ij}]+D([\gamma_{ij}])).
\label{sista}
\end{equation}
Now, for each degree in the series expansion of $\gamma_{ij}$, $D(\cdot)$ acts as multiplication
operator, and therefore equation $(\ref{sista}) \equiv 0  \pmod {r^2}$ is equivalent to
$ u_1^i u_1^j \gamma_{ij}+ z^2 [\gamma_{ij}] \equiv 0 \pmod {r^2}$. However, $\gamma_{ij}$
is composed as the sum of the matrices $B_1, B_2, \ldots, B_6$ in Lemma \ref{nygdeco},
and it is  easily checked that $ u_1^i u_1^j (B_k)_{ij}+ z^2 [(B_k)_{ij}] \equiv 0 \pmod {r^2}$
for $k=1,2, \ldots 6$, and therefore Lemma \ref{reddim1} follows.
\end{proof}

\end{document}